\documentclass[journal,twoside,web]{ieeecolor}
\usepackage{generic}
\usepackage{amsmath,amssymb,amsfonts}

\usepackage{cite}
\usepackage{algorithmic}
\usepackage{graphicx}
\usepackage{textcomp}

\usepackage{graphics} 
\usepackage{epsfig} 
\usepackage{color,soul}
\usepackage{balance}

\usepackage{acronym}

\newtheorem{theorem}{Theorem}[section]
\newtheorem{lemma}[theorem]{Lemma}
\newtheorem{proposition}[theorem]{Proposition}
\newtheorem{corollary}[theorem]{Corollary}

\newtheorem{definition}{Definition}

\newtheorem{remark}[theorem]{Remark}

\acrodef{MMSE}{minimum mean square error}

\def\BibTeX{{\rm B\kern-.05em{\sc i\kern-.025em b}\kern-.08em
    T\kern-.1667em\lower.7ex\hbox{E}\kern-.125emX}}

\begin{document}

\title{
Innovation-Based Remote State Estimation Secrecy with no Acknowledgments
}

\author{Justin M. Kennedy~\IEEEmembership{Member, IEEE},
Jason J. Ford,
\\
Daniel E. Quevedo~\IEEEmembership{Fellow, IEEE}
and
Falko Dressler~\IEEEmembership{Fellow, IEEE}
\thanks{This work has been supported in part by the project NICCI2 funded by the German Research Foundation (DFG) under grant numbers DR 639/23-2 and QU 437/1-2.}%
\thanks{JK, JF, and DQ acknowledge continued support from the Queensland University of Technology (QUT) through the Centre for Robotics.}%
\thanks{J. M. Kennedy, J. J. Ford, and D. E. Quevedo are with the School of Electrical Engineering and Robotics, Queensland University of Technology, 2 George St, Brisbane QLD, 4000 Australia. F. Dressler is with the School of Electrical Engineering and Computer Science, TU Berlin, Germany.
        {\tt\small \{j12.kennedy, j2.ford, daniel.quevedo\}@qut.edu.au, dressler@ccs-labs.org}}%
}

\maketitle

\begin{abstract}
Secrecy encoding for remote state estimation in the presence of adversarial eavesdroppers is a well studied problem.
Typical existing secrecy encoding schemes rely on the transmitter's knowledge of the remote estimator's current performance.
This performance measure is often shared via packet receipt acknowledgments.
However, in practical situations the acknowledgment channel may be susceptible to interference from an active adversary, resulting in the secrecy encoding scheme failing.
Aiming to achieve a reliable state estimate for a legitimate estimator while ensuring secrecy, we propose a secrecy encoding scheme without the need for packet receipt acknowledgments.
Our encoding scheme uses a pre-arranged scheduling sequence established at the transmitter and legitimate receiver.
We transmit a packet containing either the state measurement or encoded information for the legitimate user.
The encoding makes the packet appear to be the state but is designed to damage an eavesdropper's estimate.
The pre-arranged scheduling sequence and encoding is chosen psuedo-random.
We analyze the performance of our encoding scheme against a class of eavesdropper, and show conditions to force the eavesdropper to have an unbounded estimation performance.
Further, we provide a numerical illustration and apply our encoding scheme to an application in power systems.
\end{abstract}


\section{INTRODUCTION}

\IEEEPARstart{T}{he} emerging network of cyber-physical systems has been acknowledged as a vulnerability \cite{Sandberg2015IEEEControlSystemsCyberphysicalSecurityNetworked} with recent incidents drawing attention to the need to improve the security of these systems \cite{Tidy2022PredatorySparrowWho}.
Ensuring security of cyber-physical systems can be enhanced through control-theoretic approaches including through the utilization of the dynamics in the design \cite{Ferrari2021SafetySecurityPrivacy}.
Three key security problems exist: ensuring confidentiality of state information and control actions, integrity of transmissions, and availability of data over a network \cite{Ishii2022SecurityResilienceControl}.
In this article we focus on the problem of confidential state estimation of a remote process over an unreliable wireless network in the presence of an eavesdropper.

While the interest in control systems design is the closed-loop system performance, it is first necessary to ensure the quality of the state estimate used in the controller.
Sharing state information at every time instant provides valuable information to a legitimate user.
However, as transmissions can be intercepted by an adversarial eavesdropper, which could use transmitted state information to design future attacks on the system \cite{Sandberg2022AnnualReviewofControlRoboticsandAutonomousSystemsSecureNetworkedControl}, it becomes necessary to obfuscate the shared state estimate from an eavesdropper.
This motivates private remote state estimation techniques to ensure state secrecy.
Recent works have shown that through careful scheduling of transmissions \cite{Tsiamis2017IFACPapersOnLineStateEstimationSecrecy,Leong2019IEEETransactionsonAutomaticControlTransmissionSchedulingRemote,Liu2022AutomaticaRolloutapproachsensor} or by encrypting the packets \cite{Lucke2022IEEETransactionsonAutomaticControlCodingsecrecyremote,Tsiamis2018StateSecrecyCodes,Tsiamis2020IEEETransactionsonAutomaticControlStateSecrecyCodes,Huang2021AutomaticaEncryptionschedulingremote}, significant reduction in adversary performance can be achieved at the cost of modest degradation in legitimate user performance.
We shall explore this trade-off in our design.

Many state secrecy techniques require knowledge of the legitimate estimator's current performance, often shared through acknowledgment of successful packet receipt.
Scheduling the next transmission from the legitimate user's last received packet can be used to create the illusion of a random transmission policy to an eavesdropper \cite{Leong2019IEEETransactionsonAutomaticControlTransmissionSchedulingRemote}.
By relying heavily on acknowledgments, \cite{Tsiamis2020IEEETransactionsonAutomaticControlStateSecrecyCodes} showed that an encoded transmission of relative measurement, the innovation between the current state and the last acknowledged packet, diverges an eavesdropper's estimate.
In the case of a \emph{critical event} where the eavesdropper misses a packet that the legitimate receiver obtains, the eavesdropper is unable to recover the state estimate, and its estimation error will constantly grow\footnote{The eavesdropper's estimation error will grow to infinity in the case of unstable dynamics \cite{Tsiamis2020IEEETransactionsonAutomaticControlStateSecrecyCodes} or to the open loop estimation error in the case of stable dynamics \cite{Tsiamis2018StateSecrecyCodes}.}.
Effectively, the use of innovations with acknowledgments, can provide so-called infinite secrecy of the state information.

However, in many practical situations an acknowledgment channel may be unavailable due to hardware limitations, such as for small power limited sensors \cite{Li2008IEEESensorsJournalHybridMicropowerSource,Hoang2018IEEETransactionsonCommunicationsCellFreeMassive}, or an adversary jamming the network \cite{Ding2021IEEETransactionsonAutomaticControlRemoteStateEstimation}.
Under the encoding scheme of \cite{Tsiamis2020IEEETransactionsonAutomaticControlStateSecrecyCodes}, for an eavesdropper to maintain knowledge of the state, the adversary needs to prevent the critical event from occurring.
An active eavesdropper that combines both eavesdropping and acknowledgment blocking tasks simultaneously, such as considered in \cite{Ding2021IEEETransactionsonAutomaticControlRemoteStateEstimation}, could block all acknowledgments or be stealthy and only block acknowledgments when the critical event occurs thus hiding in the stochastic nature of the network.
This style of selective acknowledgment jamming attack has been demonstrated in practice on a carrier sense multiple access/collision avoidance (CSMA/CA) protocol WiFi network \cite{Klingler2019PosterAbstractJamming}.
While it may be possible to detect a stealthy eavesdropper using statistical methods \cite{Kennedy2022BayesianQuickestChange}, the innovation state secrecy encoding scheme of \cite{Tsiamis2020IEEETransactionsonAutomaticControlStateSecrecyCodes} is no longer secret to this powerful eavesdropper.

As further background to our current work, we note that to improve an eavesdropper's performance, an adversary could exploit the vulnerability in packet acknowledgments, including through fake acknowledgment transmission \cite{Liu2022AutomaticaRolloutapproachsensor} and event-based acknowledgment attacks \cite{Cheng2020IEEETransactionsonAutomaticControlEventBasedStealthy}.
Noting that often significant energy is required to block a communication channel, \cite{Zhang2018IEEETransactionsonControlofNetworkSystemsDoSAttackEnergy} proposed a strategy to balance performance degradation of the legitimate estimator with limited energy usage of the adversary.
An active attack to damage the legitimate user's estimate is to transmit packets that appear, in a statistical sense, to be the state measurement \cite{Naha2021QuickestDetectionDeception,Naha2022QuickestDetectionDeceptionParismonious} alongside malicious control actions \cite{Griffioen2021IEEETransactionsonAutomaticControlMovingTargetDefense}.
This has motivated watermarking schemes \cite{Naha2022SequentialDetectionReplay} and online statistical analysis \cite{Naha2022IEEETransactionsonAutomaticControlSequentialdetectionReplay} to ensure integrity of packets.
In particular, to combat the denial of service attacks in large scale power networks, \cite{Yang2022AutomaticaAdaptivedistributedKalman} proposed a distributed Kalman filter for state estimation, while \cite{Gusrialdi2022ResilientHierarchicalNetworked} used the structured nature of the system to design a cooperative control approach.

\begin{figure}
\centering
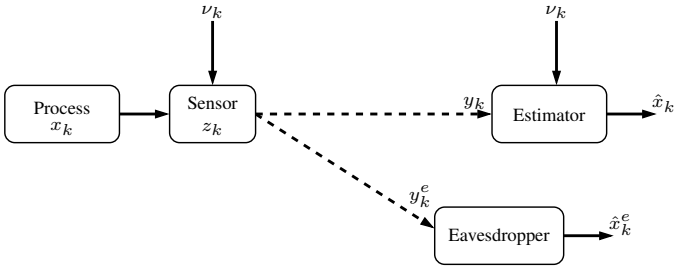
\caption{Architecture of channel environment. A remote process sends state information over an unreliable network that can be received by the legitimate estimator and eavesdropper. The packet $z_k$ is encoded with scheduling sequence $\nu_k$ which is known exactly to the legitimate estimator but not the eavesdropper. The encoding does not rely on packet receipt acknowledgments.}
\label{fig:architecture}
\end{figure}

In the present work, we are motivated to consider the problem of private remote state estimation without requiring receipt acknowledgments.
The network architecture is visualized in Figure~\ref{fig:architecture}.
To damage an eavesdropper we propose an encoding scheme, where we randomly transmit either the state measurement or a random value that has the same statistical properties as the state.
Inside the packet, we encode state information in the form of a single step innovation to improve performance of the legitimate user.
The random encoding sequence is pre-arranged between the transmitter and legitimate user.
Our proposal ensures the legitimate estimator has bounded performance and the system state remains secret to an eavesdropper.
By increasing the proportion of encoded packets to unencoded packets, the secrecy against an eavesdropper is increased at the cost of legitimate estimator performance.
We present our proposed encoding scheme in the sense of this trade-off.
Our work is inspired by the no acknowledgment secrecy scheme of \cite{Tsiamis2017IFACPapersOnLineStateEstimationSecrecy}, the use of some encrypted packets in \cite{Huang2021AutomaticaEncryptionschedulingremote}, the innovation only encoding of \cite{Tsiamis2020IEEETransactionsonAutomaticControlStateSecrecyCodes}, and the design of packets that force estimator divergence \cite{Naha2021QuickestDetectionDeception,Naha2022QuickestDetectionDeceptionParismonious}.

\emph{Summary of contributions:} We consider the problem of remote state estimation in the presence of eavesdroppers, and derive an encoding scheme to ensure secrecy of the state.
\begin{enumerate}
    \item In contrast to many recent secrecy encoding schemes, such as \cite{Tsiamis2020IEEETransactionsonAutomaticControlStateSecrecyCodes} and \cite{Huang2021AutomaticaEncryptionschedulingremote}, we consider the network environment of no packet receipt acknowledgments.
    \item We improve on \cite{Tsiamis2017IFACPapersOnLineStateEstimationSecrecy}, by transmitting encoded state information that also damages the eavesdropper.
    \item We derive expressions for the expected estimation error covariances as a function of the dynamics, channel quality, and encoding scheme.
    \item We propose an offline designed scheduling sequence to achieve a suitable measure of state secrecy.
\end{enumerate}

The remainder of the paper is structured as follows:
we present the remote estimation scenario and pose our problem in Section~\ref{sec:problem}.
In Section~\ref{sec:encoding} we propose our encoding scheme and in Section~\ref{sec:est} give the performance for the legitimate estimator. 
In Section~\ref{sec:eavesdropper} we analyze the impact of our encoding scheme on a class of eavesdropper, and in Section~\ref{sec:design} provide guidance on scheduling design and numerical results.
We illustrate application of our technique to a remote state estimation problem on a microgrid power system in Section~\ref{sec:app}.
Finally, we provide concluding remarks in Section~\ref{sec:conc}.

\section{REMOTE STATE ESTIMATION WITH AN EAVESDROPPER}
\label{sec:problem}
In this section we outline the dynamics and network model that we consider, and define the estimation of the legitimate estimator and eavesdropper.

\subsection{System Dynamics}
Consider a discrete-time LTI process with state $x_k \in \mathbb{R}^n$
\begin{equation}
    x_{k+1} = A x_k + w_k
    \label{eq:dynamics}
\end{equation}
where
$w_k$ is a zero mean Gaussian distributed process with covariance $Q$,
and $A$ is marginally stable or unstable with at least one eigenvalue on or outside of the unit circle, respectively.
Remote state estimation of unstable processes in the presence of eavesdroppers is an active problem, see for example \cite{Tsiamis2020IEEETransactionsonAutomaticControlStateSecrecyCodes,Naha2022IEEETransactionsonAutomaticControlSequentialdetectionReplay,Park2019IEEETransactionsonAutomaticControlStealthyAdversariesUncertain,Abdi2019IEEEInternetofThingsJournalPreservingPhysicalSafety}.
Additionally, many physical processes such as vehicle position or power systems \cite{Bordons2020ModelPredictiveControl} can be described with integrator models and are then marginally stable processes.
For some cyber-physical processes the control unit and actuators may be separate from the sensors \cite{Wu2022IEEETransactionsonControlofNetworkSystemsSecureControlCyber,Abdi2019IEEEInternetofThingsJournalPreservingPhysicalSafety}, which under failure could result in open-loop operation, motivating estimation of marginally stable and unstable process.

We assume that the pair $(A,\sqrt{Q})$ is controllable,
the initial state of the process $x_0$ is a Gaussian random variable with zero mean and covariance $\Sigma_0$, and that the two covariances $Q$ and $\Sigma_0$ are positive definite.
Additionally, we consider that $w_k$ and $x_0$ are uncorrelated, and $w_k$ and $w_\ell$ for $k \neq \ell$ are uncorrelated.
Finally, we assume that properties of the process $A$, $Q$, and $\Sigma_0$ are public but the realization of the state $x_0$ and noise $w_k$ are not known.

\begin{remark}
    The state $x_k$ in \eqref{eq:dynamics} could represent the state estimate from a Kalman filter using noisy measurements operating at the sensor.
    The process noise $w_k$ would then represent the Kalman innovation.
    See for example \cite{Tsiamis2020IEEETransactionsonAutomaticControlStateSecrecyCodes}.
\end{remark}

\subsection{Channel Model}
Following the standard packet based transmission utilized in network control problems, we consider that the sensor transmits a packet of state information, $z_k \in \mathbb{R}^p$, over an unreliable channel to the legitimate estimator.
The packets transmitted over the network can also be received by an eavesdropper.
To ensure privacy and secrecy of the state information, the packet is encoded.
We propose our encoding scheme in Section~\ref{subsec:legitencoding}.

We consider a particularly challenging network situation where the packet receipts by the legitimate estimator are not acknowledged to the transmitting sensor.
The lack of acknowledgment channel could be due to cost and energy usage \cite{Li2008IEEESensorsJournalHybridMicropowerSource,Hoang2018IEEETransactionsonCommunicationsCellFreeMassive}, or due to a powerful eavesdropper interfering \cite{Ding2021IEEETransactionsonAutomaticControlRemoteStateEstimation} which has been demonstrated in practice \cite{Klingler2019PosterAbstractJamming}.
An encoding scheme that relies exclusively on acknowledgments, such as \cite{Tsiamis2020IEEETransactionsonAutomaticControlStateSecrecyCodes}, may be rendered ineffective by acknowledgment blocking.
While it is possible to detect such a powerful eavesdropper \cite{Kennedy2022BayesianQuickestChange}, an alternative technique that does not rely on acknowledgments to ensure privacy should be employed.

As there are no acknowledgments, the sensor does not have knowledge of the estimation performance of the legitimate estimator.
Active privacy techniques which rely on knowledge of the remote estimator, such as scheduling \cite{Leong2019IEEETransactionsonAutomaticControlTransmissionSchedulingRemote} or encoding \cite{Tsiamis2020IEEETransactionsonAutomaticControlStateSecrecyCodes}, are unsuitable here.
Our proposed encoding scheme utilizes a pre-arranged scheduling sequence $\nu_k$ and encoding noise $\chi_k$, known to the transmitting sensor and the legitimate estimator but is unknown to the eavesdropper.
The encoding information is shared to the legitimate transmitter and eavesdropper at system initialization, isolated from an adversary. 
Information security schemes commonly use pre-arranged security information in transmission encoding \cite{Yang2020Automaticaencodingmechanismsecrecy}, cryptographic encryption \cite{Murguia2020IEEETransactionsonAutomaticControlSecurePrivateImplementation,Gao2022IEEETransactionsonAutomaticControlFaultTolerantConsensus}, and watermarking \cite{Zhou2022IEEETransactionsonAutomaticControlWatermarkingBasedProtection,Naha2021QuickestDetectionDeception}.
The challenge in our work, is the design of encoding mechanism of the state information and the design of the scheduling sequence.
Our network architecture is visualized in Figure~\ref{fig:architecture}.


Let us define $\gamma_k \in \{0,1\}$ as an indicator variable denoting successful packet reception at the legitimate estimator where

\begin{equation}
    \gamma_k = \begin{cases} 1, \mbox{ if the packet is received,} \\ 0, \mbox{ if a packet dropout occurs,} \end{cases}
    \label{eq:networkmodeldef}
\end{equation}
and similarly $\gamma_k^e \in \{0,1\}$ for outcomes at the eavesdropper.
We assume that the channel outcomes for the estimator and eavesdropper are independent and identically distributed (i.i.d.), and that the channel outcomes are independent to the initial state of the process and the process noise.
We define the channel qualities as a Bernoulli random variables where for the legitimate estimator $\mathbb{P}(\gamma_k = 1) = \mu$ and for the eavesdropper $\mathbb{P}(\gamma_k^e = 1) = \mu_e$, where $0 \leq \mu, \mu_e\leq 1$.
This model follows from standard wireless communication channels with block-fading over the channel links
\cite{Wong2011FundamentalsWirelessCommunication}.

\subsection{Minimum Mean Square Error Estimation}
The estimates of the legitimate estimator and the eavesdropper depend on information available at each time in the received packets and knowledge of the scheduling sequence $\nu_k$.
Let us define the measurements as
\begin{align*}
        y_k &= \gamma_k z_k, \mbox{ at the legitimate estimator, and}\\
        y_k^e &= \gamma_k^e z_k, \mbox{ at the eavesdropper.}
\end{align*}
We define information available to the legitimate estimator at time $k$ as 
$\mathcal{I}_k = \{\gamma_0, y_0, \nu_0, \gamma_1, y_1, \nu_1, \chi_1, \dots, \gamma_k, y_k, \nu_k, \chi_k \}$
and $\mathcal{I}_k^e = \{\gamma_0^e, y_0^e, \dots, \gamma_k^e, y_k^e\}$ for the eavesdropper. 
The legitimate estimator has perfect knowledge of the scheduling sequence $\nu_k$ and encoding noise $\chi_k$, while the eavesdropper has no knowledge.
The legitimate estimator's \ac{MMSE} estimate and associated covariance are defined as
\begin{equation}
    \hat{x}_{k|k} = E[x_k|\mathcal{I}_k], ~ P_{k|k} = E[(x_k - \hat{x}_{k|k})(x_k - \hat{x}_{k|k})^\mathsf{T}|\mathcal{I}_k] 
    \label{eq:legitMMSE}
\end{equation}
and for the eavesdropper
\begin{equation}
    \hat{x}_{k|k}^e = E[x_k|\mathcal{I}_k^e], ~ P_{k|k}^e = E[(x_k - \hat{x}_{k|k}^e)(x_k - \hat{x}_{k|k}^e)^\mathsf{T}|\mathcal{I}_k^e] .
    \label{eq:eavesdropperMMSE}
\end{equation}

\section{RANDOMIZED INNOVATION BASED ENCODING}
\label{sec:encoding}

In this section we pose the secrecy requirements, propose our encoding scheme, and decoder for the legitimate estimator.
In the following sections, we show the expected legitimate estimator performance, and provide guidance on encoding design choice for state secrecy against a class of eavesdropper.

\subsection{State Secrecy}
\label{subsec:statesecrecy}
Our objective is to design an encoding scheme that produces a reliable state estimate at the legitimate estimator using no information of the current performance, while simultaneously ensuring poor estimation performance of an eavesdropper in the network.
We introduce two notions of secrecy using the expectation of the \ac{MMSE} of the legitimate estimator and the eavesdropper.

As our encoding scheme is designed with no information of the legitimate estimator's current estimate, we do not aim for optimal estimation at every packet receipt.
Instead, we aim to ensure that the legitimate estimator's expected performance is upper bounded.
To ensure secrecy of the state estimate we seek to design the encoding scheme such that an eavesdropper's expected performance is larger than the legitimate estimator's performance.

\begin{definition}[Relative Secrecy]
\label{definition:relativesecrecy}
An encoding scheme achieves relative secrecy if and only if both the following conditions hold.
\begin{enumerate}
    \item[(i)] There exists an $\Omega > 0$ such that the trace of the legitimate estimator's \ac{MMSE} performance is upper bounded $\textrm{trace~} E[P_{k|k}] < \Omega$ for all time $k>0$.
    \item[(ii)] The trace of the \ac{MMSE} of the eavesdropper is strictly larger than that of the legitimate estimator $\textrm{trace~} E[P_{k|k}] < \textrm{trace~} E[P_{k|k}^e]$ for all time $k>0$.
\end{enumerate}

\end{definition}

We recall the definition of perfect secrecy from \cite{Tsiamis2017IFACPapersOnLineStateEstimationSecrecy} to ensure that the eavesdropper's expected estimation error diverges to infinity, while the legitimate estimator's performance remains upper bounded.

\begin{definition}[Perfect Secrecy]
\label{definition:perfectsecrecy}
An encoding scheme achieves perfect secrecy if and only if both of the following conditions hold.
\begin{enumerate}
    \item[(i)] There exists an $\Omega > 0$ such that the trace of the legitimate estimator's \ac{MMSE} performance is upper bounded $\textrm{trace~} E[P_{k|k}] < \Omega$ for all time $k>0$.
    \item[(ii)] The eavesdropper's expected \ac{MMSE} is unbounded, or equivalently the trace diverges to infinity $\textrm{trace}~ E[P_{k|k}^e] \rightarrow \infty $ as $k \rightarrow \infty$.
\end{enumerate}
\end{definition}

Definition~\ref{definition:perfectsecrecy} is stronger than Definition~\ref{definition:relativesecrecy}, as the diverging condition $\textrm{trace}~ E[P_{k|k}^e] \rightarrow \infty$ is more strict than the bounded condition $\textrm{trace~} E[P_{k|k}] < \textrm{trace~} E[P_{k|k}^e]$.

For a remote state estimator of an unstable system always transmitting the state estimate over an unreliable wireless network, i.e. $z_k = x_k$ for all $k>0$, \cite{XuEstimationuncontrolledcontrolled} showed that to ensure a bounded estimation error covariance, the network channel needs to satisfy
\begin{equation}
    1-\mu < \frac{1}{\rho(A)^2}, ~ \textrm{and} ~ 1-\mu_e < \frac{1}{\rho(A)^2}
    \label{eq:networkmodelupperbound}
\end{equation}
where $\rho(\cdot)$ is the spectral radius\footnote{The spectral radius of a matrix is defined as the maximum absolute eigenvalue $\rho(M) = \max_i |\lambda_i (M) | $ where $\lambda_i$ is the $i$th eigenvalue of $M$.}.
In this work, we assume that the channel qualities of both the legitimate estimator and eavesdropper satisfy \eqref{eq:networkmodelupperbound}, and as such are sufficient to produce bounded estimation error covariance of an unstable process in the case that the state is always transmitted.
We also do not restrict ourselves to the case $\mu_e < \mu$ as considered in \cite{Tsiamis2017IFACPapersOnLineStateEstimationSecrecy}.
To achieve state secrecy, including to cause an eavesdropper to have an unbounded estimation error covariance, we are unable just to transmit the state estimate at every time instance, motivating our encoding scheme.

\subsection{Encoding Mechanism}
\label{subsec:legitencoding}
Our proposed encoding scheme for the packet $z_k$ is
\begin{equation}
    z_k = \begin{cases}
        x_k, ~ \quad \qquad \qquad \qquad \nu_k = 0 \\
        x_k - A x_{k-1} + \chi_k, \quad \nu_k = 1
    \end{cases}
    \label{eq:encondingscheme}
\end{equation}
for $k \geq 1$ and $z_0 = x_0$,
where the sensor transmits either the current state or a single step innovation with relation to the previous state encoded by additive noise $\chi_k$.
In each packet, we only transmit the information, not the encoding values, making decoding challenging for a potential adversary.
The encoding $\nu_k$ and $\chi_k$ are pre-arranged at the sensor and legitimate estimator.
We design the scheduling sequence $\nu_k$ and encoding noises $\chi_k$ such that each packet $z_k$ appears, at least in a statistical sense, to be the current state value $x_k$.

To balance legitimate estimator estimation performance with state secrecy against eavesdroppers, several partial encoding schemes have been proposed \cite{Huang2021AutomaticaEncryptionschedulingremote,Yang2020Automaticaencodingmechanismsecrecy,Lucke2022IEEETransactionsonAutomaticControlCodingsecrecyremote}.
These transmission schemes balance providing a reliable estimate to the legitimate estimator encoding while obfuscating from an adversary.
As the legitimate estimator estimation performance can decrease, such as from a reduction in shared information \cite{Tsiamis2017IFACPapersOnLineStateEstimationSecrecy}, quantization encoding \cite{Huang2021AutomaticaEncryptionschedulingremote} or adversary attacks \cite{Lucke2022IEEETransactionsonAutomaticControlCodingsecrecyremote}, it is often necessary to send the actual state value in some of the packets.
The challenge becomes to cleverly balance the trade-off in the encoding scheme, between estimation performance and state secrecy against an eavesdropper.

In our encoding scheme, we propose that the scheduling sequence $\nu_k$ and additive encoding noise $\chi_k$ for $k \geq 1$ are chosen to be random.
The probability distributions and pseudo-random seeds are shared between the transmitter and legitimate estimator in initialization, while the eavesdropper has no knowledge.
As such, the realization of the sequence of $\nu_k$ and $\chi_k$ become pre-arranged, and are known exactly to the transmitter and legitimate estimator but unknown to an eavesdropper.
The idea of a pre-arranged distribution seed or common encoding key has been 
commonly used in several information security facets, such as in watermarking \cite{Zhou2022IEEETransactionsonAutomaticControlWatermarkingBasedProtection,Naha2021QuickestDetectionDeception} and cryptographic encryption with public and private keys \cite{Murguia2020IEEETransactionsonAutomaticControlSecurePrivateImplementation,Gao2022IEEETransactionsonAutomaticControlFaultTolerantConsensus}
as well as applications to transmission encoding schemes \cite{Yang2020Automaticaencodingmechanismsecrecy}.

We choose the distribution of the additive encoding noise $\chi_k$ to be a zero-mean Gaussian random variable with covariance
\begin{equation}
    E[\chi_k \chi_k^\mathsf{T}] = A^k \Sigma_0 (A^k)^\mathsf{T} + \sum_{\ell = 0}^{k-2} A^{k-1-\ell} Q (A^{k-1-\ell})^\mathsf{T} ,
    \label{eq:additivenoisecovariancedesign}
\end{equation}
and $\chi_k$ is uncorrelated from the process $x_0$ and $w_k$ for all $k$, and $\chi_k$ and $\chi_\ell$ for $k\neq\ell$ are uncorrelated.
Under this choice of distribution
\footnote{See Appendix~\ref{app:encodingnoise} for derivation.}
the first and second moments of each packet $z_k$ are equivalent to the state $x_k$
\begin{align*}
    E[z_k] &= E[x_k] \textrm{~and~} 
    E[z_k z_k^\mathsf{T}] = E[x_k x_k^\mathsf{T}] ,
\end{align*}
for $k \geq 1$.
An eavesdropper performing a statistical test using the first or second moment would be unable to identify whether each packet $z_k$ is the state $x_k$ or something else.
An eavesdropper directly using the packet $z_k$ as the state estimate would have a poor estimate of the true process state.

We choose the distribution of the scheduling sequence $\nu_k$ to be a Bernoulli random variable with probability of sending the state as
\begin{equation*}
    \mathbb{P}(\nu_k = 0) = \mu_d,
\end{equation*}
and encoded innovation as $\mathbb{P}(\nu_k = 1) = 1-\mu_d$.
The probability $\mu_d$ is a design variable of our encoding scheme, which trades off nominal estimation performance of the legitimate estimator where the state is sent at every time instance, against secrecy of state information.
In the case $\mu_d = 1$ only the state measurement $\nu_k=0$ is transmitted, while in the case $\mu_d = 0$ only innovations are sent.
We bound $\mu_d$ between these extremes, $0 < \mu_d < 1$, such that some of the transmissions are innovations and some are the state.
We provide guidance in Section~\ref{subsec:controlpolicy} on choice of $\mu_d$ in relation to our notions of secrecy and the expected estimation error covariance of the legitimate estimator and eavesdropper.

\section{STATE ESTIMATION PERFORMANCE}
\label{sec:est}
In this section, we present the expected performance of the legitimate estimator's state estimate.

\subsection{State Estimator}
The \ac{MMSE} estimate of the state is defined in \eqref{eq:legitMMSE} as the expectation of the state given the information received.
We define the \ac{MMSE} prediction of the state as the expectation of the state given the information available at the previous time instance
\begin{equation*}
    \hat{x}_{k|k-1} = E[x_k|\mathcal{I}_{k-1}]
\end{equation*}
where the estimate follows the dynamics,
with associated covariance as
\begin{align*}
    P_{k|k-1} &= E[(x_k - \hat{x}_{k|k-1})(x_k - \hat{x}_{k|k-1})^\mathsf{T}|\mathcal{I}_{k-1}] .
\end{align*}

From \cite[Chapter 2]{Anderson1979OptimalFiltering} the state estimate update equation is
\begin{equation}
    \hat{x}_{k|k} = \hat{x}_{k|k-1} + \gamma_k \Sigma_{k,xz} \left(\Sigma_{k,zz}\right)^{-1} \left(z_k - \hat{z}_k \right) 
    \label{eq:estupdate}
\end{equation}
with associated estimate covariance update
\begin{equation}
    P_{k|k} = \Sigma_{k,xx} - \gamma_k \Sigma_{k,xz} \left( \Sigma_{k,zz} \right)^{-1} \Sigma_{k,zx}
    \label{eq:covupdate}
\end{equation}
where the expected packet is $\hat{z}_k = E[z_k | \mathcal{I}_{k-1}]$ and the auxiliary variables are
\begin{equation*}
    \textrm{Cov} \left( \left. \begin{bmatrix}x_k \\ z_k \end{bmatrix} \right| \mathcal{I}_{k-1} \right) = \begin{bmatrix} \Sigma_{k,xx} & \Sigma_{k,xz} \\ \Sigma_{k,zx} & \Sigma_{k,zz} \end{bmatrix}
\end{equation*}
where $\Sigma_{k,xx} = P_{k|k-1}$, \\ $\Sigma_{k,zz} = E[(z_k - \hat{z}_k)(z_k - \hat{z}_k)^\mathsf{T}|\mathcal{I}_{k-1}]$, and
\begin{equation*}
    \Sigma_{k,xz} = \Sigma_{k,zx}^\mathsf{T} = E[(x_k - \hat{x}_{k|k-1})(z_k - \hat{z}_k)^\mathsf{T}|\mathcal{I}_{k-1}] .
\end{equation*}
As the pre-arranged scheduling sequence $\nu_k$ and additive noise $\chi_k$ are known to the legitimate estimator and the transmitter, the legitimate estimator's expected packet is
\begin{equation}
    \hat{z}_k = \begin{cases}
        E[x_k | \mathcal{I}_{k-1}], ~ \quad \qquad \qquad \qquad \nu_k=0 \\
        E[x_{k} - A x_{k-1} | \mathcal{I}_{k-1}] + \chi_k , \quad \nu_k=1
    \end{cases} .
    \label{eq:expectedpacketvalue}
\end{equation}

The \ac{MMSE} state estimate of the legitimate estimator is 
\begin{equation*}
    \hat{x}_{k|k} = \begin{cases}
    A \hat{x}_{k-1|k-1}, \quad \qquad \qquad \textrm{when} \quad \gamma_k = 0\\
    x_k, \qquad \qquad \qquad \qquad ~ \textrm{when} \quad (\gamma_k,\nu_k) = (1,0)\\
    x_k - A(x_{k-1} - \hat{x}_{k-1}), ~ \textrm{when} ~ (\gamma_k,\nu_k) = (1,1)
    \end{cases} .
\end{equation*}

The following theorem gives the covariance in the three possible outcomes: a dropout occurs, the state is successfully received, and an innovation is successfully received.
\begin{theorem}
\label{thm:mmseestimate}
The covariance of the legitimate estimator's state estimate is
\begin{equation*}
    P_{k|k} = \begin{cases}
    A P_{k-1|k-1} A^\mathsf{T} + Q, \quad \textrm{when} \quad \gamma_k = 0 \\
    0, \qquad \qquad \qquad \qquad ~ \textrm{when} \quad (\gamma_k,\nu_k) = (1,0) \\
    A P_{k-1|k-1} A^\mathsf{T}, \qquad \quad \textrm{when} \quad (\gamma_k,\nu_k) = (1,1)
    \end{cases} .
\end{equation*}
\end{theorem}
The proof is direct through application of the dynamics \eqref{eq:dynamics} and definition of the expectation operator \cite{Anderson1979OptimalFiltering}.

\begin{proof}
    See Appendix~\ref{app:legitmmse}.
\end{proof}

Inspecting the estimation error covariance of the legitimate estimator in Theorem~\ref{thm:mmseestimate}, we observe the following.

In the case that the transmission is dropped, $\gamma_k=0$, the estimation error covariance is the prediction error covariance.
This is the worst case as the estimation error builds by a factor of $\rho(A^2)$ and linearly by the driving noise $Q$.

In the case that the innovation is received $(\gamma_k,\nu_k) = (1,1)$, the estimation error grows by a factor of $\rho(A^2)$, which provides more information about the current state than a dropout.
Where the previous state is known exactly and $P_{k-1|k-1} = 0$, then the estimation error covariance remains zero.
The encoded innovation provides useful information to the legitimate receiver while the random additive hinders a potential eavesdropper.

Finally, in the case that the current state is received $(\gamma_k,\nu_k) = (1,0)$, the estimation error of the current state, $\hat{x}_{k|k}$, is zero as the state is received exactly.

\begin{remark}
\ac{MMSE} state estimate and associated covariances can alternatively be derived using the Kalman filter.
In the case of noisy measurements, the Kalman filter can be employed to provide similar \ac{MMSE} estimates.
While the above result apply in principle, the estimation error covariance in the case of a state receipt would not reduce to exactly zero due to the presence of measurement noise.
\end{remark}

\subsection{Expected Performance}
\label{subsec:legitexpected}
The estimation error of the legitimate estimator given in Theorem~\ref{thm:mmseestimate}, is dependent on the dynamics and the information available at the current time step $\mathcal{I}_k$: scheduling sequence $\nu_k$, additive noise $\chi_k$, and the dropout channel $\gamma_k$.
We utilize the stochastic properties of the channel environment and scheduling sequence to give the expectation of the legitimate estimator estimation error performance.

To ensure a bounded estimation error covariance at the legitimate estimator, the probability of receiving the state estimate $\mathbb{P}(\gamma_k=1,\nu_k=0) = \mu \mu_d$ must be bounded
\begin{equation}
    1 - \mu \mu_d < \frac{1}{\rho (A)^2} .
    \label{eq:decisionvariablestabilitybound}
\end{equation}
Given a channel quality $\mu$ and dynamics $A$, the minimum choice of design variable is
\begin{equation}
    \frac{1}{\mu} \left(1 - \frac{1}{\rho(A)^2}\right) < \mu_d .
    \label{eq:decisionvariableminbound}
\end{equation}
For $\mu_d < 1$, a better quality channel than \eqref{eq:networkmodelupperbound}, where only the state estimate is transmitted, is required.

In the case that the choice of $\mu_d$ does not satisfy \eqref{eq:decisionvariablestabilitybound} then $E[ P_{k|k}]$ is unbounded,
otherwise we have the following result.
\begin{theorem}
The expected estimation error covariance of the legitimate estimator's state $\hat{x}_{k|k}$ as $k\rightarrow\infty$ is
\begin{align*}
    E[P_{k|k}] = (1-\mu) S
\end{align*}
where the choice of $\mu_d$ satisfies \eqref{eq:decisionvariableminbound}, and $S$ is the unique stabilizing solutions to the Lyapunov Equation
\begin{align}
    S &= \left(\sqrt{1-\mu\mu_d} {A} \right) S \left({A}^\mathsf{T} \sqrt{1-\mu\mu_d} \right) + Q . \label{eq:solnlyap}
\end{align}
\label{thm:expectedlegitest}
\end{theorem}

\begin{proof}
    See Appendix~\ref{app:legitexpectedaug}.
\end{proof}

The proof of Theorem~\ref{thm:expectedlegitest} is shown by considering the expectation of $P_{k|k}$ as the sum of the conditional expectation of $P_{k|k}$ given the outcomes from Theorem~\ref{thm:mmseestimate} by the probability of that outcome.
Each conditional expectation of $P_{k|k}$ can be written as a function of the expectation of the previous estimation error covariance $P_{k-1|k-1}$ by application of Theorem~\ref{thm:mmseestimate}.
Expanding from time $k-1$ to the initial time $k=0$, the expectation of $P_{k|k}$ can be written as a function of the initial condition $\Sigma_0$ and a sum to time $k$ of the dynamics $A$ and $Q$ and outcome probabilities, comprised of the channel quality $\mu$ and design variable $\mu_d$.
As $k\rightarrow\infty$, we observe that the expression can be written as a converging Lyapunov equation providing the form given in Theorem~\ref{thm:expectedlegitest}.

An alternative proof approach is to consider that as $k\rightarrow\infty$, the outcomes of the legitimate estimator form a countably infinite Markov Chain (MC), see 
Appendix~\ref{app:alternativelegitexpectedaug}.
From every state in the MC, the estimation error covariance will return to a zero state when the state estimate is successfully received, such that all states in the MC are reachable.
The expectation of $P_{k|k}$ is then the sum of all of the possible MC states multiplied by the limiting distribution of the MC, or the probability of being in a state.

Inspecting the result of Theorem~\ref{thm:expectedlegitest}, we observe that the expectation of the estimation error covariance of the legitimate estimator is a function of the dynamics $A$ and $Q$, the channel $\mu$, and the encoding scheme with design variable $\mu_d$.
The performance is reduced compared to the nominal remote state estimator that transmits the state estimate every time instance, or the case that $\mu_d = 1$.
Our encoding scheme trades this nominal performance of only sending the state, with secrecy of the state information.
Using knowledge of the channel quality and dynamics, the design variable $\mu_d$ can be tuned to achieve a certain level of expected estimation error while also ensuring a bounded state estimate.
We provide guidance on our encoding design $\mu_d$ for secrecy against an eavesdropper in Section~\ref{subsec:controlpolicy} to balance performance of the legitimate estimator against secrecy to an eavesdropper.

\section{EAVESDROPPER ESTIMATION PERFORMANCE}
\label{sec:eavesdropper}
We pose our secrecy encoding scheme in the context of a class of adversarial eavesdropper that does not have knowledge of the encoding scheme.
The class of eavesdropper directly uses any packets that it \emph{believes} are the state.
This amounts to \emph{correctly} using the state in the case $\nu_k=0$, but \emph{incorrectly} using an encoded innovation as the state in the case $\nu_k=1$.
We limit our analysis to this class of eavesdropper, as even in the situation that an adversary was aware that the innovation was encoded in some of the packets, without knowledge of the additive noise $\chi_k$ the eavesdropper would be unable to extract and utilize the innovation.

As the packets are statistically equivalent to the state process, in the sense of the first and second moments, we pose three types of eavesdropper.
We consider: a naive eavesdropper that assumes every received packet is the state; a suspicious eavesdropper that suspects not every packet is the state, and has a random chance at guessing the packet type; and a smart eavesdropper that has perfect packet identification, and correctly uses the state and discards the innovation.

In this section, we show the expectation of the estimation error covariance of the class of eavesdropper, and for each type of eavesdropper compare to the legitimate user's performance.
We then provide an approach to choose an appropriate design variable $\mu_d$, and a numerical illustration.

\subsection{Expected Eavesdropper Estimation Performance}
\label{subsec:eavesexpectedresults}
At the receipt of each packet $z_k$, we consider that the eavesdropper may perform a test on the packet to make a choice whether to utilize or discard the packet.
Let us define $b_k=1$ as the case where the eavesdropper identifies a received packet $z_k$ as the state and uses the packet, and $b_k=0$ as the case where the eavesdropper identifies a received packet $z_k$ as not the state and so discards the packet.
Let us define the eavesdropper's belief to use a packet as the posterior probability test conditioned on the received packet as $\mathbb{P} (b_k=1 | z_k , \gamma_k^e, \nu_k)$, and the belief to discard a packet as $\mathbb{P} (b_k=0 | z_k , \gamma_k^e, \nu_k) = 1 - \mathbb{P} (b_k=1 | z_k , \gamma_k^e, \nu_k)$.
We outline in the following sections how each type of eavesdropper forms these conditional probabilities.

An eavesdropper has five possible events: successfully receives a state which it \emph{correctly} uses $(\gamma_k^e,\nu_k,b_k)=(1,0,1)$ or incorrectly discards $(\gamma_k^e,\nu_k,b_k)=(1,0,0)$, successfully receives an innovation which it \emph{incorrectly} uses $(\gamma_k^e,\nu_k,b_k)=(1,1,1)$ or correctly discards $(\gamma_k^e,\nu_k,b_k)=(1,1,0)$, or the packet is dropped $(\gamma_k^e=0)$.
As discarding a successfully received packet (cases $b_k=0$) is equivalent to dropping the packet $(\gamma_k^e=0)$, the five events reduce to three outcomes.

First: successfully receiving a state which the eavesdropper correctly uses, with probability
\begin{align*}
    p_r^e &= \mathbb{P}(\gamma_k^e=1,\nu_k=0,b_k=1) . 
\end{align*}
Second: successfully receiving an innovation which the eavesdropper incorrectly uses as the state, with probability
\begin{align*}
    p_i^e &= \mathbb{P}(\gamma_k^e=1,\nu_k=1,b_k=1) . 
\end{align*}
Third: the eavesdroppers drops the packet or discards a successfully received packet which it believes is not the state, with probability
\begin{align*}
    p_d^e &= \mathbb{P}(\gamma_k^e=0) + \mathbb{P}(\gamma_k^e=1,\nu_k=0,b_k=0) \\&\quad+ \mathbb{P}(\gamma_k^e=1,\nu_k=1,b_k=0) .
\end{align*}

The state estimate of an eavesdropper is
\begin{equation*}
    \hat{x}_k^e = \begin{cases}
    A \hat{x}_{k-1}^e , \quad \textrm{when} \quad (\gamma_k^e = 0) \textrm{~or~} (\gamma_k^e,\nu_k,b_k) = (1,0,0) \\ \qquad \qquad \textrm{~or~} (\gamma_k^e,\nu_k,b_k) = (1,1,0) \\
    x_k, \quad \textrm{when} \quad (\gamma_k^e,\nu_k, b_k) = (1,0,1) \\
    x_k - A x_{k-1} + \chi_k, ~~ \textrm{when} ~~ (\gamma_k^e,\nu_k,b_k) = (1,1,1)
    \end{cases}
\end{equation*}
where the predicted estimate uses dynamics, and a successfully received packet is used directly.
We derive the covariance of the state estimate similar to Theorem~\ref{thm:mmseestimate}, then follow a similar argument as Theorem~\ref{thm:expectedlegitest} for the expectation of the estimation error covariance.

\begin{lemma}
\label{lemma:classeavesdropperstateest}
The eavesdropper's estimation error covariance is
\begin{equation*}
    P_{k|k}^e = \begin{cases}
    A P_{k-1|k-1}^e A^\mathsf{T} + Q, \quad \textrm{when} \quad (\gamma_k^e = 0) \\ \qquad \qquad \textrm{~or~} (\gamma_k^e,\nu_k,b_k) = (1,0,0)\\ \qquad \qquad  \textrm{~or~} (\gamma_k^e,\nu_k,b_k) = (1,1,0) \\
    0, \quad \textrm{when} \quad (\gamma_k^e,\nu_k,b_k) = (1,0,1) \\
    2 \left(A^k \Sigma_0 (A^k)^\mathsf{T} + \sum_{\ell = 0}^{k-2} A^{k-1-\ell} Q (A^{k-1-\ell})^\mathsf{T} \right),\\ \qquad \quad \textrm{when} \quad (\gamma_k^e,\nu_k,b_k) = (1,1,1)
    \end{cases}
\end{equation*}
\end{lemma}
The proof is direct through application of the dynamics \eqref{eq:dynamics}, definition of the expectation operator \cite{Anderson1979OptimalFiltering}, and the encoding scheme \eqref{eq:encondingscheme}.
\begin{proof}
    See Appendix~\ref{app:eavesestimate}.
\end{proof}

Critically, while we can quantify in Lemma~\ref{lemma:classeavesdropperstateest} the estimation error covariance of an eavesdropper using knowledge of the mismatch between the encoding scheme and the eavesdropper's assumption, this would be unknown to the eavesdropper.
The eavesdropper assumes that upon receiving a packet $(\gamma_k^e=1)$ and utilizing the packet $(b_k=1)$ their estimation error covariance is zero, which would not be the case upon receiving an innovation.
From Lemma~\ref{lemma:classeavesdropperstateest}, we note that upon receipt and use of an innovation, the estimation error covariance is instead a function of the dynamics and time $k$.

\begin{theorem}
    \label{thm:classeavesdropperexpectation}
    The expectation of the estimation error covariance of the eavesdropper is 
    \begin{align*}
        &E[P_{k|k}^e] = (p_d^e)^k A^{k-1} \Sigma_0 (A^{k-1})^\mathsf{T} 
        + \sum_{\ell=0}^{k-1} (p_d^e)^{\ell+1} A^\ell Q (A^\ell)^\mathsf{T} \\ 
        + &p_i^e 2 \sum_{\ell=0}^{k-1} (p_d^e)^\ell \Bigl(A^k \Sigma_0 (A^k)^\mathsf{T} +  \sum_{j=0}^{k-\ell-2} A^{k-1-j} Q (A^{k-1-j})^\mathsf{T} \Bigr)
    \end{align*}
    where $p_i^e$ is the probability of receiving and utilizing an innovation, and $p_d^e$ is the probability dropping or discarding a packet.
\end{theorem}

\begin{proof}
    See Appendix~\ref{app:classeavesestimateperformance}.
\end{proof}

Inspecting the result of Theorem~\ref{thm:classeavesdropperexpectation}, we note that the expectation of the eavesdropper's estimation error covariance is a function of the dynamics $A$ and $Q$, the initial state covariance $\Sigma_0$, the time $k$, and the probability of incorrectly using an innovation $p_i^e$ and probability of dropping or discarding a packet $p_d^e$.
The probability of use or discard of encoded innovation packets depend on the belief that a received packet is the state.
We now consider three types of eavesdropper that have different packet analysis techniques and utilize the result of Theorem~\ref{thm:classeavesdropperexpectation} to compare to the legitimate estimator's performance in the sense of our definitions of secrecy.

\subsection{Naive Eavesdropper}
\label{subsec:naiveeavs}
Consider a naive eavesdropper that assumes that every packet transmitted to the legitimate estimator is the state, $\hat{z}_k = x_k$ for all $k$.
Performing basic statistical tests, such as computing the first or second moment on each received packet $z_k$, the naive eavesdropper would be unable to identify a difference between state packets $(\nu_k=0)$ and innovation packets $(\nu_k=1)$, as by design $E[z_k] = E[x_k]$, see Section~\ref{subsec:legitencoding}.
The eavesdropper's belief whether to use a packet that is successfully received is
\begin{equation*}
    \mathbb{P}(b_k=1 | z_k, \gamma_k^e=1, \nu_k) = 1 ,
\end{equation*}
irrespective of the packet containing the state $(\nu_k=0)$ or innovation $(\nu_k=1)$.
The probability of the naive eavesdropper using the state or innovation are then
\begin{align*}
    p_r^e &= \mathbb{P}(\gamma_k^e=1,\nu_k=0,b_k=1) = \mu_e \mu_d \\
    p_i^e &= \mathbb{P}(\gamma_k^e=1,\nu_k=1,b_k=1) = \mu_e (1-\mu_d)
\end{align*}
and probability of packet drop or discard is $p_d^e = 1-\mu_e$.

We state the estimation error performance of the naive eavesdropper from the result in Theorem~\ref{thm:classeavesdropperexpectation}.

\begin{corollary}
    \label{corr:naiveeavesexpectation}
    The expectation of the estimation error covariance of the naive eavesdropper diverges, $E[P_{k|k}^e] \rightarrow \infty$ as $k\rightarrow\infty$, satisfying condition (ii) of Definition~\ref{definition:perfectsecrecy}.
\end{corollary}
\begin{proof}
    See Appendix~\ref{app:naivesusexpected}.
\end{proof}

As the naive eavesdropper treats all received packets as the state, it will inadvertently use the innovation packets which significantly degrades the naive eavesdropper's state estimate.
The result of Corollary~\ref{corr:naiveeavesexpectation} gives that the expectation of the estimation error covariance is a function of time $k$ with some terms diverging as $k$ increases.
We note that the diverging terms in the expected performance are larger for larger probabilities $p_i^e$, or smaller $\mu_d$.
Thus while any choice of $\mu_d$ that satisfies \eqref{eq:decisionvariablestabilitybound} will ensure perfect secrecy, a smaller $\mu_d$ will provide faster divergence of the naive eavesdropper's estimate.
Additionally, we observe that even in the case of a perfect channel $\mu_e=1$ and $p_d^e=0$, the naive eavesdropper's expected performance still diverges.

\subsection{Suspicious Eavesdropper}
\label{subsec:suspiciouseavesdropper}
Consider an eavesdropper that becomes suspicious that not all of the received packets are the state.
This suspicious eavesdropper applies a statistical test to each packet that it receives to form a belief of whether to use the packet or to discard.
Such analysis could be performed by testing the sequence of received packet $\mathcal{I}_k^e$, using online statistical techniques such as Quickest Change Detection \cite{Kennedy2022BayesianQuickestChange,Naha2021QuickestDetectionDeception}.

As this eavesdropper is performing a statistical test on the content of the received packet $z_k$, the posterior probability to use the packet would be correlated with the value of that packet and thus the encoding $\nu_k$ and $\chi_k$.
For simplicity in analysis, we assume that the eavesdropper has a fixed random chance of correctly identifying a packet upon receipt, independent of the packet value, encoding, or previous test outcome.
As such, our assumption is that the probability of belief is i.i.d. and uncorrelated from the packet $z_k$.
While a major simplifying assumption, this permits the below result, which gives an indication to the potential eavesdropper performance in the situation of non-perfect statistical tests.
In the following section, we analyze a smart eavesdropper that has perfect detection through statistical analysis of received packets.

Let us define the probability for the eavesdropper to use a packet that contains the state as
\begin{equation*}
    \mathbb{P}(b_k=1|z_k, \gamma_k^e = 1 ,\nu_k=0) = \mu_b,
\end{equation*}
and to use a packet that contains the innovation as
\begin{equation*}
    \mathbb{P}(b_k=1|z_k, \gamma_k^e = 1 ,\nu_k=1) = \bar{\mu}_b 
\end{equation*}
where $0< \mu_b<1$ and $0<\bar{\mu}_b<1$.
By the independence assumption, the probability of the suspicious eavesdropper using the state or innovation are then
\begin{align*}
    p_r^e &= \mathbb{P}(\gamma_k^e=1,\nu_k=0,b_k=1) \\
    &= \mathbb{P}(\gamma_k^e=1)\mathbb{P}(\nu_k=0)\mathbb{P}(b_k=1 |z_k, \gamma_k^e = 1 ,\nu_k=0) \\
    &= \mu_e \mu_d \mu_b , \\
    p_i^e &= \mathbb{P}(\gamma_k^e=1,\nu_k=1,b_k=1) \\
    &= \mathbb{P}(\gamma_k^e=1)\mathbb{P}(\nu_k=1)\mathbb{P}(b_k=1 |z_k, \gamma_k^e = 1 ,\nu_k=1) \\
    &= \mu_e (1-\mu_d) \bar{\mu}_b ,
\end{align*}
and the probability to drop or discard a packet is
\begin{equation*}
    p_d^e = 1 - \mu_e\mu_d\mu_b -\mu_e\bar{\mu}_b + \mu_e \mu_d\bar{\mu}_b.
\end{equation*}
The above probabilities are a consequence of the assumption that the channel quality, schedule to transmit the state, and eavesdropper belief, are i.i.d. random variables and uncorrelated from each other and the process.

We state the estimation error performance of the suspicious eavesdropper from the result in Theorem~\ref{thm:classeavesdropperexpectation}.

\begin{corollary}
    \label{corr:suspiciouseavesexpectation}
    The expectation of the estimation error covariance of the suspicious eavesdropper diverges, $E[P_{k|k}^e] \rightarrow \infty$ as $k\rightarrow\infty$, satisfying condition (ii) of Definition~\ref{definition:perfectsecrecy}.
\end{corollary}

\begin{proof}
    See Appendix~\ref{app:naivesusexpected}.
\end{proof}

As the suspicious eavesdropper has a random chance of incorrectly identifying encoded innovation packets as the state, it will inadvertently use these packets which significantly degrades its state estimate.
The result of Corollary~\ref{corr:suspiciouseavesexpectation} gives that the expectation of the estimation error covariance is a function of time $k$ with some terms diverging as $k$ increases.

In contrast to the naive eavesdropper's expected estimation error covariance, see Corollary~\ref{corr:naiveeavesexpectation}, the probability of using an innovation, $p_i^e$, is smaller
$\mu_e (1-\mu_d) \geq \mu_e (1-\mu_d) \bar{\mu}_b$, for $\bar{\mu}_b < 1$, but the probability of the dropout, $p_d^e$, is much larger.
For some choices of the eavesdropper's beliefs $\mu_b$ and $\bar{\mu}_b$ the expectation of the suspicious eavesdropper's performance will be worse than for the naive eavesdropper.

\begin{remark}
Consider the scenario where the suspicious eavesdropper correctly identifies all packets that contain the state, such that $\mu_b = 1$ but makes errors on the innovation packets such that $\bar{\mu}_b > 0$ and $p_i^e > 0$. 
By Corollary~\ref{corr:suspiciouseavesexpectation} the expectation of the eavesdropper's estimation error covariance will diverge.
We note that for errors in identification of the encoded innovations such that the eavesdropper uses these packets will diverge the eavesdropper's estimate.
\end{remark}


\subsection{Smart Eavesdropper}
\label{subsec:smarteaves}
Consider a smart eavesdropper that analyses the packets,
but in contrast to the suspicious eavesdropper has perfect performance.
The smart eavesdropper perfectly identifies all received packets that are the state measurement 
\begin{equation*}
    \mathbb{P}(b_k=1| z_k, \gamma_k^e = 1,\nu_k=0) = 1 ,
\end{equation*}
and uses these packets.
The smart eavesdropper perfectly identifies all received packets that are not the state
\begin{equation*}
    \mathbb{P}(b_k=1 | z_k, \gamma_k^e = 1,\nu_k=1) = 0 ,
\end{equation*}
and discards these packets.
Effectively, the smart eavesdropper can identify the sequence $\nu_k$.
However, we consider that it does not know the realization of $\chi_k$ and  is unaware of the full encoding mechanism \eqref{eq:encondingscheme}, so cannot decode the innovations.
We consider that it would be challenging for an eavesdropper to identify the value of $\chi_k$ inside the packet $z_k$ as the random variable is independent and uncorrelated from the state process $x_k$ and scheduling sequence $\nu_k$.

The probability of the smart eavesdropper using the state or innovation is
\begin{align*}
    p_r^e &= \mathbb{P}(\gamma_k^e=1,\nu_k=0,b_k=1) = \mu_e\mu_d, \\
    p_i^e &= \mathbb{P}(\gamma_k^e=1,\nu_k=0,b_k=1) = 0,
\end{align*}
and probability of packet drop or discard is
\begin{equation*}
    p_d^e = 1 - \mu_e \mu_d .
\end{equation*}
We note that the second outcome introduced in Section~\ref{subsec:eavesexpectedresults} is eliminated.
We note that this is the best type of eavesdropper in the class that we analyze.
For an eavesdropper to obtain better performance, an adversary would need to decode the innovation, which is outside of the class that we consider.

The smart eavesdropper effectively functions as a remote state estimator where the state is transmitted every time instance with a channel quality of $p_r^e = \mu_e \mu_d$.
This result is a consequence of our encoding scheduling sequence $\nu_k$ being i.i.d. and uncorrelated to the eavesdropper's channel.
Following \cite{XuEstimationuncontrolledcontrolled}, a necessary and sufficient condition for the smart eavesdropper to have a bounded estimation error covariance,
is that the encoding design probability is upper bounded by
\begin{equation}
    \frac{1}{\mu_e} \left(1 - \frac{1}{\rho(A)^2} \right) < \mu_d .
    \label{eq:eavesdroppernetworkcorrect}
\end{equation}

The result of Lemma~\ref{lemma:classeavesdropperstateest} can be reduced by noting that the case $(\gamma_k^e,\nu_k,b_k) = (1,1,1)$  is discarded.
The state estimate of the smart eavesdropper is
\begin{equation*}
    \hat{x}_k^e = \begin{cases}
    A \hat{x}_{k-1}^e , \quad \textrm{when~} \gamma_k^e = 0 \textrm{~or~} (\gamma_k^e,\nu_k) = (1,1) \\ \qquad \textrm{~or~} (\gamma_k^e,\nu_k,b_k) = (1,0,0) \\
    x_k, \quad \textrm{when} \quad (\gamma_k^e,\nu_k,b_k) = (1,0,1) 
    \end{cases}
\end{equation*}
with covariance
\begin{equation*}
    P_{k|k}^e = \begin{cases}
    A P_{k-1|k-1}^e A^\mathsf{T} + Q, \\ \qquad \textrm{when} \quad \gamma_k^e = 0 \textrm{~or~} (\gamma_k^e,\nu_k) = (1,1) \\ \qquad \textrm{~or~} (\gamma_k^e,\nu_k,b_k) = (1,0,0) \\
    0, \quad \textrm{when} \quad (\gamma_k^e,\nu_k) = (1,0) 
    \end{cases}
\end{equation*}
This can be shown directly from \cite{Anderson1979OptimalFiltering} and is simpler than the state estimate of the legitimate estimator, see Theorem~\ref{thm:mmseestimate}.
Unlike the naive and suspicious eavesdroppers above, the smart eavesdropper can correctly quantify its own estimation error covariance, $P_{k|k}^e$, as it is aware of the nature of the packets it is using.

To compare the performance of the smart eavesdropper to the legitimate estimator, we establish the expectation of the estimation error covariance of the smart eavesdropper.
In the case that $\mu_e$ or $\mu_d$ do not satisfy \eqref{eq:eavesdroppernetworkcorrect}, then $E[P_{k|k}^e]$ is unbounded.
In the case that $\mu_e$ and $\mu_d$ satisfy \eqref{eq:eavesdroppernetworkcorrect} then we have the following result.
\begin{lemma}
The expectation of the estimation error covariance of the smart eavesdropper as $k\rightarrow\infty$ is
\begin{equation*}
    E[P_{k|k}^e] = (1 - \mu_e \mu_d) S^e
\end{equation*}
where $\mu_e$ and $\mu_d$ satisfy \eqref{eq:eavesdroppernetworkcorrect}, and $S^e$ is the unique stabilizing solution to the Lyapunov Equation
\begin{equation}
    S^e = \left(\sqrt{1 - \mu_e \mu_d}A\right) S^e \left(A^\mathsf{T} \sqrt{1 - \mu_e \mu_d} \right) + Q .
    \label{eq:solnlyapeaves}
\end{equation}
\label{lemma:smarteavesexpectedest}
\end{lemma}
The proof follows that of Theorem~\ref{thm:expectedlegitest} and Theorem~\ref{thm:classeavesdropperexpectation}, but is simpler as the eavesdropper has only two possible channel outcomes.
As $p_i^e = 0$ for the smart eavesdropper, then the diverging terms in Theorem~\ref{thm:classeavesdropperexpectation} are removed, and the expectation then converges.

\begin{proof}
    See Appendix~\ref{app:smarteavsexpected}.
\end{proof}

To show secrecy as a function of design variable $\mu_d$ and channel qualities, $\mu$ and $\mu_e$, we give the following monotonicity result of the Lyapunov equation.
\begin{lemma}
    \label{lemma:monotonicityresult}
    Consider a $\beta,\beta^\star$ where $0 < \beta , \beta^\star < 1$, $\rho(\sqrt{1-\beta} A) < 1$ and $\rho(\sqrt{1-\beta^\star} A) < 1$
    and introduce the following two Lyapunov equations
    \begin{align*}
        W &= \sqrt{1-\beta} A W A^\mathsf{T} \sqrt{1-\beta} + Q \\
        W^\star &= \sqrt{1-\beta^\star} A W^\star A^\mathsf{T} \sqrt{1-\beta^\star} + Q 
    \end{align*}
    where $W$ and $W^\star$ are the unique stabilizing solutions.
    In the case that $\beta^\star < \beta$ then 
    \begin{equation*}
        \textrm{trace~} W < \textrm{trace~} W^\star .
    \end{equation*}
\end{lemma}
\begin{proof}
    See Appendix~\ref{app:lyapeqmonotonicity}.
\end{proof}

Using Lemmas~\ref{lemma:smarteavesexpectedest} and~\ref{lemma:monotonicityresult} and Theorem~\ref{thm:expectedlegitest}, we compare the expected estimation error of the smart eavesdropper against the legitimate estimator.
The differences in performance are related to the difference in channel qualities, and scheduling sequence design.
We explore the cases where the eavesdropper channel quality is worse than, or equal to, the legitimate estimator's channel quality.

\begin{theorem}
    In the case that the eavesdropper has a worse or equal quality channel to the legitimate estimator, $\mu_e \leq \mu$ and the scheduling sequence is chosen in the range
    \begin{equation}
        \frac{1}{\mu_e}\left(1 - \frac{1}{\rho(A)^2}\right) < \mu_d < 1 
        \label{eq:thm:smarteavaesdropperrelativeperformance:bounded}
    \end{equation}
    then the trace of the expected estimation error of the legitimate estimator is strictly less than the eavesdropper
    \begin{equation*}
        \textrm{trace~} E[P_{k|k}] < \textrm{trace~} E[P_{k|k}^e] .
    \end{equation*}
    This satisfies condition (ii) of Definition~\ref{definition:relativesecrecy}.
    \label{thm:smarteavaesdropperrelativeperformance:bounded}
\end{theorem}

\begin{proof}
    Recall \eqref{eq:decisionvariableminbound}, Theorem~\ref{thm:expectedlegitest} and Lemma~\ref{lemma:smarteavesexpectedest}.
    For any $\mu_d<1$ and $\mu_e \leq \mu$ then $1-\mu < 1-\mu_e \mu_d$.
    In the case $\mu_e = \mu$ then $S \equiv S^e$.
    In the case $\mu_e < \mu$, let $\beta=\mu\mu_d$ and $\beta^\star=\mu_e\mu_d$ and via Lemma~\ref{lemma:monotonicityresult}, $\textrm{trace~}S < \textrm{trace~}S^{e}$.
    It follows in both cases that $(1-\mu)\textrm{trace~}S < (1-\mu_e\mu_d) \textrm{trace~}S^{e}$.
\end{proof}

From Theorem~\ref{thm:smarteavaesdropperrelativeperformance:bounded}, we can conclude that our encoding scheme achieves a level of relative secrecy against the smart eavesdropper that has an equal or worse channel quality.
In the case where the eavesdropper has a strictly worse quality channel and the dynamics are unstable such that $\rho(A) > 1$, we observe an extension to Theorem~\ref{thm:smarteavaesdropperrelativeperformance:bounded}.

\begin{theorem}
    In the case that $\mu_e < \mu$, and the dynamics are unstable $\rho(A) > 1$,
    the smart eavesdropper's expected state estimate is unbounded $E[P_{k|k}^e] \rightarrow \infty$  while the legitimate estimator's estimate is bounded where the design variable $\mu_d$ is bounded by
    \begin{equation}
        \frac{1}{\mu}\left(1 - \frac{1}{\rho(A)^2} \right) < \mu_d \leq \frac{1}{\mu_e}\left(1 - \frac{1}{\rho(A)^2}\right) .
        \label{eq:thm:smarteavaesdropperrelativeperformance:unbounded}
    \end{equation}
    This satisfies condition (ii) of Definition~\ref{definition:perfectsecrecy}.
    \label{thm:smarteavaesdropperrelativeperformance:unbounded}
\end{theorem}

\begin{proof}
Choice of $\mu_d$ satisfying \eqref{eq:decisionvariableminbound} to ensure a bounded estimate for the legitimate estimator informs the lower bound.
Failing \eqref{eq:eavesdroppernetworkcorrect} such that the eavesdropper has an unbounded estimation error covariance informs the upper bound.
\end{proof}

Comparing the result of Theorem~\ref{thm:smarteavaesdropperrelativeperformance:unbounded} to the proposal of \cite{Tsiamis2017IFACPapersOnLineStateEstimationSecrecy} the bound on the random transmission of the state is similar to achieve perfect secrecy.
However, our encoder is different as we transmit an encoded innovation instead of no information, which the legitimate estimator can decode, providing better legitimate estimation performance while still ensuring secrecy of the state estimate against an eavesdropper.

Under a channel model with signal fading over distance, we might expect the case of eavesdropper channel quality worse than the legitimate estimator to be more common, as a stealthy eavesdropper might be physically located further away from the transmitter as considered in \cite{Tsiamis2017IFACPapersOnLineStateEstimationSecrecy}.

We observe from the results of Theorems~\ref{thm:smarteavaesdropperrelativeperformance:bounded} and~\ref{thm:smarteavaesdropperrelativeperformance:unbounded}, that through the use of the innovations in our encoder design, the legitimate estimator has lower expectation of estimation error covariance than a smart eavesdropper and thus a better state estimate in the case of better or the same channel quality.
Our proposed encoding technique is most beneficial in the case where the legitimate user has a better or equal channel quality to the eavesdropper.

\section{SCHEDULING SEQUENCE DESIGN FOR SECRECY}
\label{sec:design}

We now discuss approaches to determine an appropriate design variable $\mu_d$ to generate the scheduling sequence, and provide a numerical illustration.
Let us briefly recall our packet encoding from \eqref{eq:encondingscheme}
\begin{equation*}
    z_k = \begin{cases}
        x_k, ~ \quad \qquad \qquad \qquad \nu_k = 0 \\
        x_k - A x_{k-1} + \chi_k, \quad \nu_k = 1
    \end{cases}
\end{equation*}
for $k \geq 1$ and $z_0 = x_0$,
and we randomize transmission of the state with $\mathbb{P}(\nu_k = 0) = \mu_d$,
and $\chi_k$ is a zero-mean Gaussian random variable designed such that the first and second moment of the packet are the same as the state.

\subsection{Scheduling Distribution Design}
\label{subsec:controlpolicy}
Consider a given dynamics $A$ and $Q$, and legitimate estimator channel quality $\mu$ and eavesdropper channel quality $\mu_e$, then the expectations of the estimation error covariance can be written as a function of $\mu_d$.
The expectation of the estimation error covariance of the legitimate estimator from Theorem~\ref{thm:expectedlegitest} can be written as
\begin{equation}
    J(\mu_d) = \textrm{trace~} E[P_{k|k}] = (1-\mu) \textrm{trace~} S,
    \label{eq:legitimateoptimisationcostfunction}
\end{equation}
and the smart eavesdropper from Lemma~\ref{lemma:smarteavesexpectedest} 
\begin{equation*}
    J_e(\mu_d) = \textrm{trace~} E[P_{k|k}^e] = (1 - \mu_e \mu_d) \textrm{trace~} S^e
\end{equation*}
where $S$ and $S^e$ are functions of $\mu_d$, see \eqref{eq:solnlyap} and \eqref{eq:solnlyapeaves}.
Before providing a method to select an encoding design $\mu_d$, we observe the following monotonicity result.
\begin{lemma}
    The trace of the expectation of the estimation error covariance for the legitimate estimator $J(\mu_d)$ and smart eavesdropper $J_e(\mu_d)$ are monotonically decreasing in $\mu_d$, such that for $\mu_d^\star \leq \mu_d$
    \begin{align*}
        J(\mu_d) \leq J(\mu_d^\star) \quad \textrm{and} \quad J_e(\mu_d) \leq J_e(\mu_d^\star) .
    \end{align*}
    \label{lemma:monotonicity}
\end{lemma}
\begin{proof}
    Consider $\mu_d^\star < \mu_d$.
    Recall Theorem~\ref{thm:expectedlegitest}, and let $\beta=\mu\mu_d$ and $\beta^\star=\mu\mu_d^\star$ then via Lemma~\ref{lemma:monotonicityresult}, $\textrm{trace~}S < \textrm{trace~}S^{\star}$.
    Recall Lemma~\ref{lemma:smarteavesexpectedest}, and let $\beta^e=\mu_e\mu_d$, $\beta^{e,\star}=\mu_e\mu_d^\star$ then via Lemma~\ref{lemma:monotonicityresult}, $\textrm{trace~}S^e < \textrm{trace~}S^{e,\star}$, and $1-\beta < 1-\beta^\star$.
    The result follows.
\end{proof}

The result of Lemma~\ref{lemma:monotonicity} gives that as we decrease $\mu_d$ towards the minimum value in \eqref{eq:decisionvariableminbound}, and as such transmit more innovations, the expectation of the estimation error covariance of both the legitimate estimator and the smart eavesdropper increase.
Conversely as we increase $\mu_d$ towards $1$, such that we transmit fewer innovations, the expectation of the estimation error covariance of both the legitimate estimator and the smart eavesdropper reduces.
Our design variable $\mu_d$ then trades off the estimation performance of the legitimate estimator for secrecy against the eavesdropper.

We now establish a range on the encoding design $\mu_d$ to satisfy the constraints \eqref{eq:decisionvariableminbound} and $\mu_d < 1$ and condition (i) of our secrecy Definitions~\ref{definition:relativesecrecy} and~\ref{definition:perfectsecrecy}.
Applying the monotonicity result of Lemma~\ref{lemma:monotonicity}, the minimum choice of $\mu_d$ will be at the bound $J(\mu_d) = \Omega$ by a given $\Omega > 0$, while the maximum choice will be at the bound $J(\mu_d) = 0$.
The minimum choice that ensures that the expected estimation error covariance of the legitimate estimator is bounded by a given $\Omega > 0$ can be found by maximizing\footnote{Using any standard constrained nonlinear solver.} \eqref{eq:legitimateoptimisationcostfunction} over possible $\mu_d$
\begin{equation*}
    \mu_d^{\min} = \arg \max J(\mu_d) < \Omega
\end{equation*}
such that the constraints \eqref{eq:decisionvariableminbound} and $\mu_d < 1$ hold.
The maximum choice can be found by minimizing greater than 0
\begin{equation*}
    \mu_d^{\max} = \arg \min J(\mu_d) > 0  
\end{equation*}
such that the constraints \eqref{eq:decisionvariableminbound} and $\mu_d < 1$ hold.
A choice of choice in the range $\mu_d^{\min} < \mu_d < \mu_d^{\max}$ ensures condition (i) of Definitions~\ref{definition:relativesecrecy} and~\ref{definition:perfectsecrecy}.

For the case $\mu_e \leq \mu$, while any choice of encoding design $\mu_d$ in the ranges $\mu_d^{\min} < \mu_d < \mu_d^{\max}$ and \eqref{eq:thm:smarteavaesdropperrelativeperformance:bounded}, from Theorem~\ref{thm:smarteavaesdropperrelativeperformance:bounded}, will give secrecy under Definition~\ref{definition:relativesecrecy}, we may be interested in the encoding design that maximizes the secrecy gain.
To maximize the secrecy gain, we desire to find the encoding that achieves the biggest performance difference.
Let us introduce a function of the difference in expectation of estimation error covariances
\begin{align*}
    J_r(\mu_d) &= \textrm{trace~} E[P_{k|k}^e] - \textrm{trace~} E[P_{k|k}] \\
    &= (1-\mu_e \mu_d) \textrm{trace }(S^e) - (1-\mu) \textrm{trace }S
\end{align*}
where both $S^e$ and $S$ are functions of the design $\mu_d$.
We note that $J_r(\mu_d) > 0$ for any $\mu_d$ in the range \eqref{eq:thm:smarteavaesdropperrelativeperformance:bounded}, as $\textrm{trace~} E[P_{k|k}^e] >  \textrm{trace~} E[P_{k|k}]$ by Theorem~\ref{thm:smarteavaesdropperrelativeperformance:bounded}.

To obtain an encoding design $\mu_d^\star$ that maximizes the estimation error covariance difference, we find the $\mu_d$ that maximizes $J_r(\mu_d) > 0$
\begin{equation}
    \mu_d^\star = \arg \max J_r(\mu_d) > 0
    \label{eq:optimisationfunction}
\end{equation}
such that the constraints $\mu_d^{\min} < \mu_d^\star < \mu_d^{\max}$, and \eqref{eq:thm:smarteavaesdropperrelativeperformance:bounded} hold.
The choice of $\mu_d^\star$ for the encoding design will provide the biggest secrecy gain against the smart eavesdropper.

In the case of better eavesdropper channel quality $\mu < \mu_e$ there may exist a range on $\mu_d$ where our encoding design will satisfy Definition~\ref{definition:relativesecrecy}.
There may exist a value $\mu_d^\star$ that maximizes $J_r(\mu_d) > 0$ from the optimization \eqref{eq:optimisationfunction} such that only the constraint $\mu_d^{\min} < \mu_d^\star < \mu_d^{\max}$ is satisfied.
Noting that the constraint \eqref{eq:thm:smarteavaesdropperrelativeperformance:bounded} does not apply in the case $\mu < \mu_e$.
If an encoding design $\mu_d^\star$ exists, then there may also be a range $\underline{\mu}_d < \mu_d^\star < \overline{\mu}_d$ that provides $J(\mu_d) > 0$, and can be computed
\begin{align*}
    \underline{\mu}_d &= \arg \min J_r(\mu_d) > 0
\end{align*}
with constraint $\mu_d^{\min} < \underline{\mu}_d < \mu_d^\star$ and
\begin{align*}
    \overline{\mu}_d &= \arg \min J_r(\mu_d) > 0
\end{align*}
with constraint $\mu_d^\star < \overline{\mu}_d < \mu_d^{\max}$.

Finally, in the case of worse eavesdropper channel quality $\mu_e < \mu$ the best legitimate estimator performance, where the eavesdropper has an unbounded estimate under Definition~\ref{definition:perfectsecrecy}, is given by 
\begin{equation*}
    \mu_d^{\max} = \arg \min J(\mu_d) > 0  
\end{equation*}
such that the constraint \eqref{eq:thm:smarteavaesdropperrelativeperformance:unbounded} holds.

\subsection{Numerical Illustration}
\label{subsec:numerical}
We briefly illustrate the relative performance of the legitimate estimator and smart eavesdropper in a numerical example.
We do not illustrate the performance of the naive and suspicious eavesdroppers, as via Corollary~\ref{corr:naiveeavesexpectation} and~\ref{corr:suspiciouseavesexpectation} the estimation performance is divergent for any $\mu_d$.

Consider the dynamics in \eqref{eq:dynamics} with
\begin{equation*}
    A = \begin{bmatrix}1 & 0.3 \\ 0.5 & 1.001\end{bmatrix}, \textrm{~and~} Q = 10^{-3} \begin{bmatrix} 1 & 0 \\ 0 & 1 \end{bmatrix}
\end{equation*}
where we note that $\rho(A) = 1.3878 > 1$.
Consider a channel quality of $\mu = 0.9$ for the legitimate estimator.
Using \eqref{eq:decisionvariableminbound} we obtain that the design variable is lower bounded $0.5342 < \mu_d$.

Let us consider four cases of smart eavesdropper channel quality of $\mu_e^1 = 0.85$, $\mu_e^2 = \mu$, $\mu_e^3 = 0.95$, and $\mu_e^4 = 0.99$.
Figure~\ref{fig:numericalillustration} shows the absolute difference in the traces of the expected estimation error between the smart eavesdropper and the legitimate estimator  $J_r(\mu_d) / J(\mu_d)$, against the encoding design variable $\mu_d$ for the four cases.

In the case that the eavesdropper's channel quality is worse ($\mu_e^1 < \mu$ in dotted magenta) or equal ($\mu_e^2 = \mu$ in dashed blue) to the legitimate estimator, the trace of the expected estimation error covariance of the eavesdropper, while bounded, is larger for any choice of valid design.
These results illustrate Theorem~\ref{thm:smarteavaesdropperrelativeperformance:bounded} and satisfaction of Definition~\ref{definition:relativesecrecy}.

For the case where the eavesdropper has a worse quality channel $\mu_e^1 < \mu$, using \eqref{eq:thm:smarteavaesdropperrelativeperformance:unbounded} from Theorem~\ref{thm:smarteavaesdropperrelativeperformance:unbounded} encoding designs in the range $0.5342 < \mu_d < 0.5656$ force the smart eavesdropper's estimation error covariance to be unbounded while the legitimate estimator's estimation error covariance remains bounded, achieving Definition~\ref{definition:perfectsecrecy}.

In the case that the eavesdropper's channel quality is better than the legitimate estimator $\mu_e^3 = 0.95$, see the solid black line in Figure~\ref{fig:numericalillustration}, there is a visible range of $\mu_d$ where $J_r(\mu_d) > 0$.
Optimizing \eqref{eq:optimisationfunction}, the encoding design $\mu_d = 0.5745$ gives the largest positive value of $J_r(\mu_d)$, with the range $0.5342 < \mu_d < 0.8931$ giving $J_r(\mu_d)>0$.
In some scenarios where the eavesdropper has a better quality channel, our encoding design can provide relative secrecy under Definition~\ref{definition:relativesecrecy}.

For a near perfect eavesdropper channel of $\mu_e^4 = 0.99$, a choice of $\mu_d$ that provides $J_r(\mu_d) > 0$ is not apparent in Figure~\ref{fig:numericalillustration} (dot-dashed red line).
Using \eqref{eq:optimisationfunction}, the encoding design $\mu_d = 0.5571$ gives the largest positive value of $J_r(\mu_d)$, and the range $0.5342 < \mu_d < 0.9384$ gives $J_r(\mu_d)>0$.
Even in the scenario where an eavesdropper has a significantly better quality channel, our encoding design still provides relative state secrecy under Definition~\ref{definition:relativesecrecy}.

\begin{figure}
    \centering
    \includegraphics[width=0.45\textwidth]{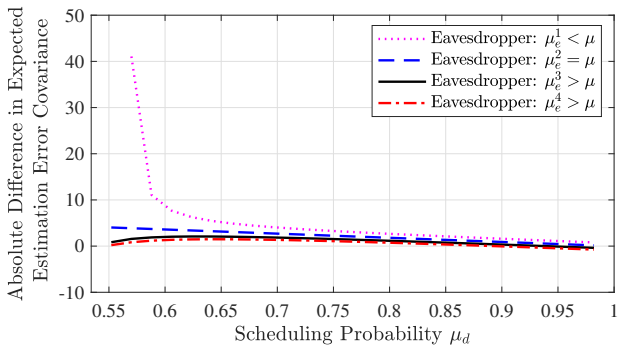}
    \caption{Comparison of the absolute difference in trace of the expected estimation error covariance of the legitimate estimator compared with the smart eavesdropper with four channel qualities (worse, equal, better, much better). Eavesdropper with worse channel quality in dotted magenta, equal channel quality in dashed blue, better channel quality in solid black, and much better channel quality in dot-dashed red. The results of Theorem~\ref{thm:smarteavaesdropperrelativeperformance:bounded} are apparent where the eavesdropper has worse performance than the legitimate estimator in the case of worse or equal channel quality.}
    \label{fig:numericalillustration}
\end{figure}

\section{APPLICATION TO POWER SYSTEMS}
\label{sec:app}

We now consider an application of our proposed transmission encoding scheme to a microgrid.
A microgrid is a small electricity grid, typically consisting of local generation, such as solar photo-voltaics, and local storage, such as batteries, to supply a small to medium load.
The load could include a typical suburban house, several houses, or contained facility such as a hospital \cite{Bordons2020ModelPredictiveControl}.
In metropolitan areas, the microgrid can connect to the main grid with import and export capability, while in remote areas, the microgrid is isolated.
The interconnection between multiple microgrids and to the main utility grid, enables coordination to achieve global system goals.
However, this networking exposes the whole power system to cyber-attacks altering the behavior of the system \cite{Gallo2020IEEETransactionsonAutomaticControlDistributedCyberAttack}.

With advancements in solar generation and battery storage technology, there has recently been a rise in the microgrid `prosumer' \cite{Liu2017IEEETransactionsonPowerSystemsEnergySharingModel}.
The `prosumer' both produces electricity and exports to the grid, as well as consumes and imports power from the grid.
The challenge of a grid connected microgrid is to control the power flow to either maximize the use of the local storage and minimize purchase of power from the grid, or to maximize the export of power to the grid for profit \cite{Zhang2017IEEETransactionsonPowerSystemsRobustOperationMicrogrids}.

While individuals may benefit from maximizing sale of power to the grid, many users in a small geographic area exporting power can cause grid instability \cite{Olivares2014IEEETransactionsonSmartGridTrendsMicrogridControl}.
As more consumer households transition to microgrids with local power generation and storage, it becomes necessary for a network operator to monitor and control the connected microgrid to ensure stability \cite{Guerrero2013IEEETransactionsonIndustrialElectronicsAdvancedControlArchitectures}.
The transmission of consumer data, and behavior as extracted from power flow data poses a privacy risk \cite{Lisovich2010IEEESecurityPrivacyInferringPersonalInformation}.
This motivates the associated cybersecurity problem to ensure confidentiality of the storage levels and generation potential from eavesdroppers.

To autonomously control the power flows in a connected microgrid, \cite{Bordons2020ModelPredictiveControl} posed a constrained model predictive control design.
Their experimental microgrid consisted of a battery and hydrogen storage systems, green power from solar panels, household load, and a grid connection for export for sale and purchase import power.
Figure~\ref{fig:microgridillustration} illustrates the power flow connections in this example microgrid.

\begin{figure}
\centering
\resizebox{0.45\textwidth}{!}{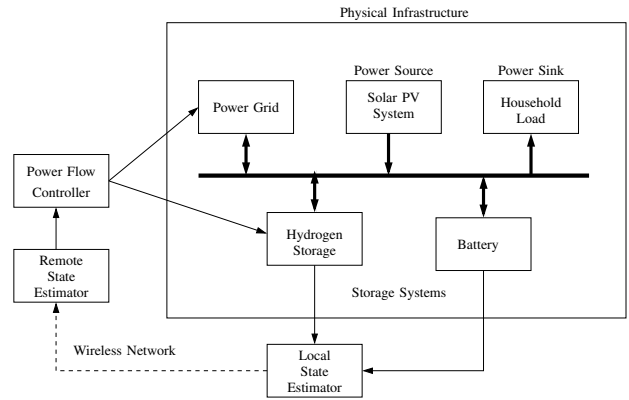}
\caption{Illustration of the microgrid power flow connections, adapted from \cite{Bordons2020ModelPredictiveControl}. Local green power supplies a small to medium sized load, such as a house, with batteries and hydrogen system providing power storage. The controller manages the power flows to maximize the use of the storage systems and minimize purchase of power from the grid.}
\label{fig:microgridillustration}
\end{figure}

\subsection{Microgrid Dynamics}
The dynamics of the battery and hydrogen storage systems can be parameterized as nonlinear ordinary differential equations.
For the purposes of control, \cite{Bordons2020ModelPredictiveControl} introduced a discrete-time linearized model to describe the change in storage charge from the input power flows.
The model states are the percentage battery state of charge ($SOC$) and hydrogen level ($LOH$) such that $x = [SOC, ~LOH]^\mathsf{T}$, the control inputs are the hydrogen power flow $P_H$ and the grid power flow $P_{grid}$ such that $u = [P_H, ~P_{grid}]$, while the green power $P_{solar}$, and load $P_{load}$, are considered uncontrolled input disturbances.
The power to the battery storage is the sum of all power flows by Kirchhoff's laws
\begin{equation*}
    P_{bat} = P_{load} - P_{solar} - P_H  - P_{grid} .
\end{equation*}
All power flows are in kW.
The discrete-time linearized dynamics posed in \cite{Bordons2020ModelPredictiveControl} are
\begin{equation}
    x_{k+1} = A x_k + B u_k + B_d (P_{solar} - P_{load})
    \label{eq:microgriddynamics}
\end{equation}
where the sampling rate is $1$ second, $A$ is the identity matrix of size $2 \times 2$ and
\begin{equation*}
    B = \begin{bmatrix} 1.56 & 1.56 \\ -5.66 & 0 \end{bmatrix} \times 10^{-3}, \quad
    B_d = \begin{bmatrix} 1.56 \\ 0 \end{bmatrix} \times 10^{-3} .
\end{equation*}
We note that the system is marginally stable.
A MATLAB/Simulink implementation of the MPC controller, nonlinear storage system models, and sample data for one $24$ hour day of solar power generation and household load used in \cite{Bordons2020ModelPredictiveControl} is available online\footnote{\texttt{http://institucional.us.es/agerar/simugrid/}}.
At the chosen sampling rate there are $86400$ data points. 

\subsection{Transmission Encoding Performance}
We extend this system by considering that the two storage systems have a one-way wireless network connection to the digital controller.
At the battery and hydrogen system, a local Kalman filter computes a state estimate to filter measurement and process noise.
This local state estimate is then the transmitted state measurement of the storage system levels.
This state estimate using the microgrid dynamics \eqref{eq:microgriddynamics} can then be written in the form \eqref{eq:dynamics},
where $A$ is the identity matrix of size $2 \times 2$ and the process noise $w_k \sim \mathcal{N}(0,Q)$ encodes the Kalman innovation and the control actions.
Through testing on the simulation the covariance of the process noise was found to be approximately $Q \approx I_2 \times 10^{-5}$ where $I_2$ is the identity matrix of size $2 \times 2$.

We perform a Monte Carlo simulation of $1000$ trials of the one day of sample generation data from \cite{Bordons2020ModelPredictiveControl} to illustrate the estimation error performance difference between the legitimate estimator and the smart eavesdropper.
We consider that the two remote estimators have the same channel qualities of $\mu = \mu_e = 0.6$, and we investigate design variable probabilities in the range $\mu_d = \{0.1, 0.9\}$.

\begin{figure}
    \centering
    \includegraphics[width=0.45\textwidth]{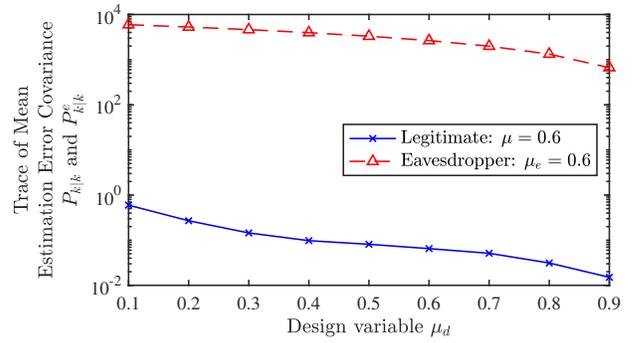}
    \caption{Monte Carlo Simulation of Microgrid with transmission encoding of remote state estimate. Eavesdropper performance is significantly reduced compared to the legitimate estimator by randomly sending true state and one step innovation.}
    \label{fig:microgridEstimationError}
\end{figure}

Figure~\ref{fig:microgridEstimationError} shows the mean of the estimation error covariances $P_{k|k}$ (solid blue) and $P_{k|k}^e$ (dashed red) across the Monte Carlo trials and across the simulation time $k$, against the design variable probability $\mu_d$.
As the proportion of the state measurement is sent increases, $\mu_d \rightarrow 1$, the estimation error decreases for both the legitimate estimator and the eavesdropper.
However, the mean estimation error for the eavesdropper is considerably larger than for the legitimate estimator, greater than $10^3$ compared to less than $10^0$.

The error in the state measurement does degrade the controller performance.
Considering the power flow to the grid connection as a measure of controller performance, as grid flow equates to power sold or purchased, we compare the total power flow over the day using our encoding scheme against no transmission encoding.
The difference in grid power flow is below $1.58\%$ for decision probability $\mu_d = 0.1$, highlighting that there is marginal impact on control performance even at the most restrictive encoding scheme.

\section{CONCLUSIONS}
\label{sec:conc}
This article investigated the problem of remote state estimation in the presence of an eavesdropper, under a challenging network environment.
We consider the situation where the transmitter and legitimate estimator receiver do not have a packet receipt acknowledgment channel.
This scenario could arise due to hardware limitations or the actions of an adversary jamming the network.

We have developed a state-secrecy code that randomly alternates between sending the state and a random value packet that appears to statistically be the state.
The random value packet both damages the eavesdropper's state estimate, while containing encoded state information for the legitimate estimator.
Our encoding scheme ensures that the legitimate estimator's expected estimation performance remains bounded.
We design our encoding to provide a measure of expected secrecy against an eavesdropper.

An open problem is to ensure state secrecy against intelligent eavesdroppers that learn the encoding scheme.


\appendix

\subsection{Encoding Scheme Additive Noise Design}
\label{app:encodingnoise}
For any finite $k>0$, the expectation of the state $x_k$ is
\begin{align*}
E[x_k] &= A^k E[ x_0] + \sum_{\ell=0}^{k-1} A^{k-1-\ell} E[w_{\ell}] 
    = 0 
\end{align*}
recalling that the initial state $x_0$ and every $w_k$ are zero mean.
As such, in case $\nu_k=0$ of sending the state $z_k = x_k$ then
\begin{equation*}
    E[z_k] = E[x_k] = 0 .
\end{equation*}
Now consider case $\nu_k=1$ with the innovation encoded by additive noise $z_k = x_k - A x_{k-1} + \chi_k$ then
\begin{align*}
    E[z_k] &= E[x_k - A x_{k-1} + \chi_k] 
    = E[w_{k-1}] + E[\chi_k] = 0 
\end{align*}
by the design that $\chi_k$ is zero-mean.

The covariance of the state can be found to be \cite{Anderson1979OptimalFiltering}
\begin{equation*}
    E[x_k x_k^\mathsf{T}] = A^k \Sigma_0 (A^k)^\mathsf{T} + \sum_{\ell=0}^{k-1} A^{k-1-\ell} Q (A^{k-1-\ell})^\mathsf{T}.
\end{equation*}
In the case $\nu_k=0$ this would be the covariance of the packet.
Consider the covariance of the packet in the case $\nu_k=1$ where the packet $z_k = x_k - Ax_{k-1} + \chi_k = w_{k-1} + \chi_k$ is the innovation encoded by additive noise 
\begin{align*}
&E[z_k z_k^\mathsf{T}] 
= E[(w_{k-1} + \chi_k)(w_{k-1} + \chi_k)^\mathsf{T}] \\
&= E[w_{k-1} w_{k-1}^\mathsf{T}] + E[w_{k-1} \chi_k^\mathsf{T}] + E[\chi_k w_{k-1}^\mathsf{T}] + E[\chi_k \chi_k^\mathsf{T}] \\
&= Q + 0 + 0 + A^k \Sigma_0 (A^k)^\mathsf{T} + \sum_{\ell=0}^{k-2} A^{k-1-\ell} Q (A^{k-1-\ell})^\mathsf{T} \\
&= A^k \Sigma_0 (A^k)^\mathsf{T} + \sum_{\ell = 0}^{k-1} A^{k-1-\ell} Q (A^{k-1-\ell})^\mathsf{T} 
= E[x_k x_k^\mathsf{T}] 
\end{align*}
where $\chi_k$ is uncorrelated from $w_\ell$ for all $k,\ell \geq 1$ and by design the covariance of $\chi_k$ is chosen as
\begin{equation*}
    E[\chi_k \chi_k^\mathsf{T}] = A^k \Sigma_0 (A^k)^\mathsf{T} + \sum_{\ell=0}^{k-2} A^{k-1-\ell} Q (A^{k-1-\ell})^\mathsf{T}.
\end{equation*}

\subsection{Legitimate Estimator MMSE}
\label{app:legitmmse}
Proof of Theorem~\ref{thm:mmseestimate}.
The proof shows the \ac{MMSE} of the state estimate for the legitimate estimator.

\begin{proof}
We consider the three outcomes separately, let us first consider a packet drop $\gamma_k = 0$.
From \eqref{eq:estupdate} and \eqref{eq:covupdate} with $\gamma_k = 0$ then $\hat{x}_{k|k} = \hat{x}_{k|k-1} = E[x_k | \mathcal{I}_{k-1}]$ where
\begin{align*}
    \hat{x}_{k|k-1} 
    &= E[A x_{k-1} | \mathcal{I}_{k-1}] + E[ w_{k-1} | \mathcal{I}_{k-1}] = A \hat{x}_{k-1|k-1} 
\end{align*}
recalling that $w_k$ is defined as a zero-mean Gaussian random variable such that $E[w_{k-1}|\mathcal{I}_{k-1}] = 0$, and $E[x_{k-1} | \mathcal{I}_{k-1}] = \hat{x}_{k-1|k-1}$, 
and the estimation error covariance is $P_{k|k} = \Sigma_{k,xx} = P_{k|k-1}$ where
\begin{align*}
    &P_{k|k-1} = E[(x_k - \hat{x}_{k|k-1})(x_k - \hat{x}_{k|k-1})^\mathsf{T}|\mathcal{I}_{k-1}] \\
    = &E[(A x_{k-1} + w_{k-1} - A \hat{x}_{k-1|k-1})\\&(A x_{k-1} + w_{k-1} - A \hat{x}_{k-1|k-1})^\mathsf{T}|\mathcal{I}_{k-1}] \\
    = &A E[(x_{k-1} - \hat{x}_{k-1|k-1}) (x_{k-1} - \hat{x}_{k-1|k-1})^\mathsf{T}|\mathcal{I}_{k-1}] A^\mathsf{T} \\&+ E[w_{k-1} w_{k-1}^\mathsf{T} |\mathcal{I}_{k-1}] \\
    = &A P_{k-1|k-1} A^\mathsf{T} + Q
\end{align*}
recalling that $w_k$ and $x_k$ are uncorrelated.

Now consider a successfully received state transmission $(\gamma_k,\nu_k) = (1,0)$ where $z_k = x_k$, and from \eqref{eq:expectedpacketvalue} the expected packet is $z_k = E[x_k | \mathcal{I}_{k-1}] = \hat{x}_{k|k-1}$.
Now with $z_k = x_k$ and $\hat{z}_k = \hat{x}_{k|k-1}$ we present the parts of \eqref{eq:estupdate} and \eqref{eq:covupdate} required 
\begin{align*}
    \Sigma_{k,zz} 
    = &E[(x_k - \hat{x}_{k|k-1})(x_k - \hat{x}_{k|k-1})^\mathsf{T}|\mathcal{I}_{k-1}] \\
    = &A P_{k-1|k-1} A^\mathsf{T} + Q, 
\end{align*}
and
\begin{align*}
    \Sigma_{k,xz} 
    = &E[(x_k - \hat{x}_{k|k-1})(x_k - \hat{x}_{k|k-1})^\mathsf{T}|\mathcal{I}_{k-1}] \\
    = & A P_{k-1|k-1} A^\mathsf{T} + Q . 
\end{align*}
Then applying the estimation update \eqref{eq:estupdate}
\begin{align*}
    &\hat{x}_{k|k} = \hat{x}_{k|k-1} + \gamma_k \Sigma_{k,xz} \left(\Sigma_{k,zz}\right)^{-1} \left(z_k - \hat{z}_k \right) \\
    &= \hat{x}_{k|k-1} + (A P_{k-1|k-1} A^\mathsf{T} + Q ) \\ &\quad\times \left( A P_{k-1|k-1} A^\mathsf{T} + Q \right)^{-1} \left(x_k - \hat{x}_{k|k-1} \right) \\
    &= x_k 
\end{align*}
with covariance \eqref{eq:covupdate}
\begin{align*}
    P_{k|k} &= \Sigma_{k,xx} - \gamma_k \Sigma_{k,xz} \left( \Sigma_{k,zz} \right)^{-1} \Sigma_{k,zx} \\
    &= (A P_{k-1|k-1}A^\mathsf{T} + Q) - ( A P_{k-1|k-1}A^\mathsf{T} + Q ) = 0.
\end{align*}

Finally, consider a successfully received innovation transmission $(\gamma_k,\nu_k) = (1,1)$ where $z_k = x_k - A x_{k-1} = w_{k-1}$, and from \eqref{eq:expectedpacketvalue} the expected packet is 
\begin{align*}
    \hat{z}_k &= E[x_k - A x_{k-1} | \mathcal{I}_{k-1}] + \chi_k \\
    &= A \hat{x}_{k-1|k-1} - A\hat{x}_{k-1|k-1} + \chi_k 
    = \chi_k .
\end{align*}
which gives the preliminary result
\begin{align*}
    &z_k - \hat{z}_k 
    = w_{k-1} + \chi_k - \chi_k 
    = w_{k-1} .
\end{align*}
Now we present the parts of \eqref{eq:estupdate} and \eqref{eq:covupdate} required using the above result
\begin{align*}
    \Sigma_{k,zz} &= E[(z_k - \hat{z}_k) (z_k - \hat{z}_k)^\mathsf{T} | \mathcal{I}_{k-1}] 
    = E[ w_{k-1}  w_{k-1}^\mathsf{T} ] 
    = Q
\end{align*}
and
\begin{align*}
    \Sigma_{k,xz} &= E[(x_k - \hat{x}_{k|k-1}) (z_k - \hat{z}_{k})^\mathsf{T} | \mathcal{I}_{k-1} ] \\
    &= E[(A x_{k-1} + w_{k-1} - A \hat{x}_{k-1|k-1}) w_{k-1}^\mathsf{T} | \mathcal{I}_{k-1} ] \\
    &= E[A(x_{k-1} - \hat{x}_{k-1|k-1}) w_{k-1}^\mathsf{T} + w_{k-1} w_{k-1}^\mathsf{T} | \mathcal{I}_{k-1}] \\
    &= 0 + Q = Q .
\end{align*}
Then the estimate \eqref{eq:estupdate} is
\begin{align*}
    &\hat{x}_{k|k} = \hat{x}_{k|k-1} + \gamma_k \Sigma_{k,xz} \left( \Sigma_{k,zz} \right)^{-1} (z_k - \hat{z}_k) \\
    &= A \hat{x}_{k-1|k-1} + Q Q^{-1} w_{k-1} \\
    &= A \hat{x}_{k-1|k-1} + w_{k-1} = x_k - A(x_{k-1} - \hat{x}_{k-1|k-1}) .
\end{align*}
where $w_{k-1} = x_{k} - A x_{k-1}$,
with covariance \eqref{eq:covupdate}
\begin{align*}
    &P_{k|k} = \Sigma_{k,xx} - \gamma_k \Sigma_{k,xz} \left(\Sigma_{k,zz}\right)^{-1} \Sigma_{k,zx} \\
    &= ( A P_{k-1|k-1} A^\mathsf{T} + Q ) -   Q  \left(Q\right)^{-1}  Q  \\
    &= A \bar{P}_{k-1|k-1} A^\mathsf{T} .
\end{align*}
This completes the proof.
\end{proof}

\subsection{Legitimate Estimator Expected Estimation Error Covariance}
\label{app:legitexpectedaug}
Proof of Theorem~\ref{thm:expectedlegitest}.
The following proof shows the expected estimation error covariance of the state at the legitimate estimator.

\begin{proof}
    We consider that the legitimate estimator is able to decode the packages that it successfully receives.
    There are then three outcomes for the legitimate estimator, 
    successful receipt of a state estimate $(\varphi_k=1)$ with probability $p_r = \mathbb{P}(\gamma_k = 1, \nu_k = 0) = \mu \mu_d$,
    successful receipt of an innovation $(\varphi_k=2)$ with probability $p_i = \mathbb{P}(\gamma_k = 1, \nu_k = 1) = \mu (1-\mu_d)$,
    and a standard dropout $(\varphi_k=3)$ with probability $p_d = \mathbb{P}(\gamma_k = 0) = (1-\mu)$.
    The expectation of the estimation error covariance for time $k>0$ can be written as a sum of the sequence of dropouts from the first transmission
    \begin{align}
        &E[P_{k|k}] = (p_i + p_d)^k A^k \Sigma_0 (A^\mathsf{T})^k \label{eq:legitexpectedinduction} \\ &+ p_d \sum_{j=0}^{k-1} (p_i + p_d)^j A^j Q (A^\mathsf{T})^j \nonumber .
    \end{align}
    
    We show via a proof by induction.
    
    Consider $k=0$ from definition $ E[\bar{P}_{0|0}] = \Sigma_0 B^\mathsf{T} p_d $, as only the state is transmitted at the first time and $\Sigma_0$ is the covariance of the initial state $x_0$.
    Now consider the first time $k=1$ from definition
    \begin{align*}
        &E[P_{1|1}] = \sum_{y=1}^3 E[P_{1|1} | \varphi_1 = y] P(\varphi_1 = y) \\
        &= (p_d + p_i) A \Sigma_0 A^\mathsf{T} + p_d Q .
    \end{align*}
    We now show that if \eqref{eq:legitexpectedinduction} holds for time $k$, that the form \eqref{eq:legitexpectedinduction} also holds for time $k+1$.
    \begin{align*}
        &E[P_{k+1|k+1}] = \sum_{y=1}^3 E[P_{k+1|k+1} | \varphi_{k+1} = y] P(\varphi_{k+1} = y) \\
        &= AE[P_{k|k}]A^\mathsf{T} p_i + E[P_{k+1|k}] p_d 0 p_r \\
        &= AE[P_{k|k}]A^\mathsf{T} p_i + \left(AE[\bar{P}_{k|k}]A^\mathsf{T} + Q \right) p_d \\
        &= Q p_d + AE[\bar{P}_{k|k}]A^\mathsf{T} (p_i + p_d) .
    \end{align*}
    Consider the expression $A E[P_{k|k}] A^\mathsf{T} (p_i + p_d)$ utilizing \eqref{eq:legitexpectedinduction} for $E[P_{k|k}]$,
    the first term
    \begin{align*}
         &A \left( (p_i + p_d)^k A^k \Sigma_0 (A^\mathsf{T})^k \right)A^\mathsf{T} (p_i + p_d) \\
         &= (p_i + p_d)^{k+1} A^{k+1} \Sigma_0 (A^\mathsf{T})^{k+1} ,
    \end{align*}
    the second term
    \begin{align*}
         &A \left( p_d \sum_{j=0}^{k-1} (p_i + p_d)^j A^j Q (A^\mathsf{T})^j \right)A^\mathsf{T} (p_i + p_d) \\
         &= p_d \sum_{j=1}^{k} (p_i + p_d)^{j} A^{j} Q  (A^\mathsf{T})^{j} .
    \end{align*}
    Now
    \begin{align*}
        &E[P_{k+1|k+1}] = (p_i + p_d)^{k+1} A^{k+1} \Sigma_0 (A^\mathsf{T})^{k+1} \\&\quad+ Q p_d + p_d \sum_{j=1}^{k} (p_i + p_d)^{j} A^{j} Q (A^\mathsf{T})^{j} \\
        &= (p_i + p_d)^{k+1} A^{k+1} \Sigma_0 (A^\mathsf{T})^{k+1}  \\&\quad+ p_d \sum_{j=0}^{k} (p_i + p_d)^{j} A^{j} Q (A^\mathsf{T})^{j} ,
    \end{align*}
    which is the form of \eqref{eq:legitexpectedinduction} at time $k+1$.
    This completes the induction argument.

    Let us explore the stabilizing solutions of the two terms of \eqref{eq:legitexpectedinduction} as $k\rightarrow\infty$.
    The first term results from a sequence of a dropouts from the initial transmission. 
    By assumption of $\mu_d$ in \eqref{eq:decisionvariablestabilitybound}, we note that $\rho(\sqrt{p_i + p_d} A) = \rho(\sqrt{1-\mu\mu_d} A) < 1$, so as time $k\rightarrow\infty$ then $(\sqrt{p_i + p_d} A)^k \rightarrow 0$ and the initial estimation error covariance $\Sigma_0$ is exponentially forgotten.

    The second term encodes the sequences of potential dropouts and innovations occurring from the first dropout after the estimator received a state packet. The sum is comprised of the potential value of the estimation error covariance multiplied by the corresponding probability. Taking as $k\rightarrow\infty$, this result can be shown with a countably infinite, irreducible, and aperiodic Markov Chain, see
    Appendix~\ref{app:alternativelegitexpectedaug}.
    %
    Consider the sum in \eqref{eq:legitexpectedinduction} from $j=0$ to $k-1$ and denote as $S_k$,
    \begin{align*}
        S_{k-1} &= \sum_{j=0}^{k-1} (p_i + p_d)^j A^j Q (A^\mathsf{T})^j \\
        &= \sum_{j=0}^{k-1} \left(\sqrt{p_i + p_d} A\right)^j  Q \left(A^\mathsf{T} \sqrt{p_i + p_d}\right)^j .
    \end{align*}
    By assumption of $\mu_d$ in \eqref{eq:decisionvariablestabilitybound}, we note that $\rho(\sqrt{p_i + p_d} A) = \rho(\sqrt{1-\mu\mu_d} A) < 1$, so the sum is a vector geometric series, and can be written in the form of a discrete-time Lyapunov equation sequence \cite{Anderson1979OptimalFiltering}
    \begin{equation*}
        S_{k} = \sqrt{p_i + p_d} A \bar{S}_{k-1}A^\mathsf{T} \sqrt{p_i + p_d} + Q
    \end{equation*}
    from $S_0 = Q$.
    The stabilizing solution $S$ can be found by taking $k\rightarrow\infty$, or setting $S_{k} = S_{k-1} = S$ and solving for the unique stabilizing solution to
    \begin{equation*}
        S = \sqrt{p_i + p_d} A S A^\mathsf{T} \sqrt{p_i + p_d}  + Q.
    \end{equation*}

    We conclude the proof by stating the expectation of the state using the above results 
    \begin{equation*}
        E[P_{k|k}] = (1-\mu) S .
    \end{equation*}
\end{proof}

\subsection{Alternative Legitimate Estimator Expected Estimation Error Covariance}
\label{app:alternativelegitexpectedaug}
Proof of Theorem~\ref{thm:expectedlegitest}.
The following proof shows the expected estimation error covariance of the state at the legitimate estimator using an alternative Markov Chain approach.

It is possible to derive the estimation error at some time $k$ with exact knowledge of the past sequence of dropouts and encoding.
To find the expectation of the estimation error covariance at some time $k$, we can take the total expectation over all possible sequences of dropouts and encoding, the summation of the final estimation error covariances weighted by the possibility of each sequence.
The proof is in two parts, we first define a Markov Chain representation of the possible sequences of packet receipts, before second computing the total expectation.

The possible sequences can be written as a resetting random walk from the `in-sync' state and the first dropout.
The `in-sync' state is where the estimation error covariance is zero from either receiving the state, or receiving an innovation after receiving the state.
On the third outcome, a dropout, the estimation error covariance diverges from zero.
Let us define $\Delta$ as the number of steps from the `in-sync' state, where $\Delta=1$ is the first dropout after being `in-sync'.
From this first dropout, the estimation error covariance can increase from a dropout or innovation for two possible states, or reset to the `in-sync' state.
From the second dropout or innovation, the estimation error covariance can increase further with a dropout or innovation for now four possible states, or reset to the `in-sync' state.
As $k\rightarrow\infty$, this resetting random walk can be written as a countably infinite Markov Chain \cite{Bremaud2020MarkovChains}, where the first state is the `in-sync' state, the second state is the first dropout, and every following state of dropouts and innovations follows.

The states of the Markov Chain represent the time since the last packet receipt that maintained the `in-sync' state.
For clarity we recall the following three probabilities
$p_r = \mathbb{P}(\gamma_k = 1, \nu_k = 0) = \mu \mu_d$ as the probability of receiving a state and resetting,
$p_i = \mathbb{P}(\gamma_k = 1, \nu_k = 1) = \mu (1-\mu_d)$ as the probability of receiving an innovation,
and $p_d = \mathbb{P}(\gamma_k = 0) = (1-\mu)$ as the probability of a packet dropout.
Additionally, we define $\pi_j$ as the limiting distribution of state $j$ of the Markov Chain with state $\pi_0$ the `in-sync' state, $\pi_1$ as the first dropout from `in-sync', and we recall that the sum of all probabilities is equal to 1: $\sum_{j=0}^{\infty} \pi_j = 1 $.
The Markov Chain is visualized in Figure~\ref{fig:MCbubble}.
This Markov Chain represents all possible estimation error covariance resulting from sequences of dropouts and packet receipts.

\begin{figure}
\centering
\resizebox{0.45\textwidth}{!}{\input{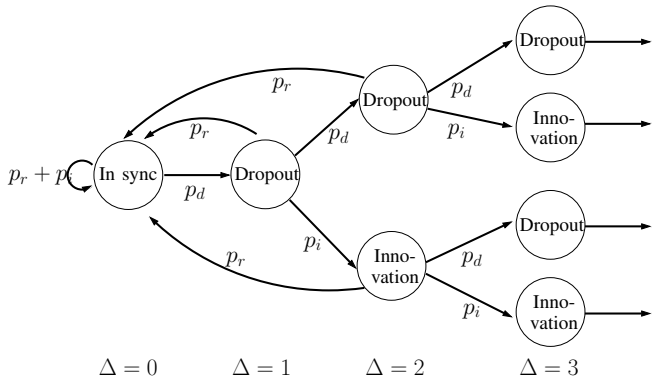}}
\caption{Markov Chain representation of the states of the legitimate estimator. The first state is the `in-sync' state where the estimate is in sync with the transmitter. The legitimate estimator remains in sync from state or innovation receipts, or drops the packet moving to the second state. The estimate after this first dropout is dependent on the sequence of further dropouts of successful packet receipts of transmitted innovations. At any point, the estimator can receive a state packet and return or reset to the `in-sync' state.}
\label{fig:MCbubble}
\end{figure}

In computing the total expectation of the estimation error covariance, we take the summation of the estimation error covariance value weighted by the probability of the state in the Markov Chain.
The limiting distribution of a Markov Chain represents the proportion of time that is spent in a given state \cite{Bremaud2020MarkovChains}.

In the following proposition we state the limiting distribution of the countably infinite Markov Chain which characterizes the possible sequences of dropouts and innovations, then give a proof.
\begin{proposition}
\label{proposition:limitingdist}
The limiting distribution of the in-sync state is
\begin{equation*}
    \pi_0 = \frac{p_r}{1-p_i} , 
\end{equation*}
the limiting distribution of the first state from one dropout where $\Delta = 1$ is
\begin{equation*}
    \pi_D = p_d \pi_0 . 
\end{equation*}
From each state there are two possible options to continue the sequence for a total of $\mathsf{N} = 2^{\Delta-1}$ at each step $\Delta$.
The limiting distribution of the first step $\Delta=2$, when an innovation is received is: $\pi_{DI} = p_i \pi_D = p_i p_d \pi_0 = \mu(1-\mu_d) (1-\mu)\pi_0$, 
or when a dropout occurs is: $\pi_{DD} = p_d \pi_D = p_d p_d \pi_0 = (1-\mu)^2 \pi_0 $.
Thus the limiting distribution of one of the $\mathsf{N}$ states at step $\Delta > 1$ from `in-sync' is based on the sequence to that state 
\begin{equation*}
    \pi_{j \ell} = p_i^j p_d^\ell p_d \pi_0  = (\mu(1-\mu_d))^j (1-\mu)^\ell (1-\mu) \pi_0
\end{equation*}
where $j \geq 0$ is the number of innovations received and $\ell \geq 0$ is the number of dropouts, and the number of innovations and dropouts at step $\Delta$ are bounded by $j + \ell + 1 = \Delta$.
\end{proposition}

\begin{proof}
Proof of Proposition~\ref{proposition:limitingdist}.

There are three options that can occur at any state: receiving the state with probability $p_r$, receiving an innovation with probability $p_i$, and a dropout with probability $p_d$.
Receiving the state measurement will return to the `in-sync' state.

In the Markov Chain the `in-sync' state, denoted $\pi_0$, can be reached from any other state when either the state is received $\nu_k=0$ or from the `in-sync' state when an innovation is received $\nu_k=1$.
We can write the transitions into $\pi_0$ as the sum of every state multiplied by the probability of receiving a state $p_r$ and the probability of an innovation $p_i$ from the state $\pi_0$
\begin{align*}
    \pi_0 &= (p_r + p_i) \pi_0 + p_r \pi_N \\&\qquad+ p_r \pi_{DI} + p_r \pi_{DD} + \dots \\
    (1 - p_r - p_i) \pi_0 &= p_r \sum_{j=1}^\infty \pi_j \\
    \sum_{j=1}^\infty \pi_j &= \frac{(1 - p_r - p_i)}{p_r} \pi_0
\end{align*}
Recall that the sum of probabilities is equal to 1 then combine with the above result
\begin{align*}
    \sum_{j=0}^\infty \pi_j &= 1 \\
    \pi_0 + \sum_{j=1}^\infty \pi_j &= 1 \\
    \pi_0 + \frac{(1 - p_r - p_i)}{p_r} \pi_0 &= 1 \\
    \pi_0 \frac{1 - p_i}{p_r} &= 1 \\
    \pi_0 &= \frac{p_r}{1 - p_i}
\end{align*}
which shows the first part of the proposition.

When in the `in-sync' state, receiving a state or innovation will ensure the state remains in the `in-sync' state, and can only leave with the first occurrence of dropout.
Thus the limiting distribution of the first dropout state is
\begin{equation}
    \pi_D = p_d \pi_0 .
    \label{eq:limiting:firstdrop}
\end{equation}
From the first dropout state $\Delta=1$, receiving the state will return to the `in-sync' state but receiving an innovation or dropout will propagate the error to the next step $\Delta=2$ with the following $\mathsf{N}=2$ limiting distributions
\begin{align*}
    \pi_{DI} &= p_i \pi_D = p_i p_d \pi_0 \\
    \pi_{DN} &= p_d \pi_D = p_d p_d \pi_0
\end{align*}
by application of \eqref{eq:limiting:firstdrop} for the step $\Delta=2$.
To the next step $\Delta=3$, receiving the state will return to the `in-sync' state again, but receiving an innovation or dropout will propagate the error to the following $\mathsf{N}=4$ limiting distributions
\begin{align*}
    \pi_{DII} &= p_i \pi_{DI} = p_i p_i p_d \pi_0 = p_i^2 p_d \pi_0 \\
    \pi_{DID} &= p_d \pi_{DI} = p_d p_i p_d \pi_0 = p_i p_d^2 \pi_0 \\
    \pi_{DDI} &= p_i \pi_{DD} = p_i p_d p_d \pi_0 = p_i p_d^2 \pi_0 \\
    \pi_{DDD} &= p_d \pi_{DD} = p_d p_d p_d \pi_0 = p_d^3 \pi_0
\end{align*}
by application of \eqref{eq:limiting:firstdrop}, noting that the limiting distribution is comprised of the number of innovations and dropouts from the limiting distribution of the `in-sync' state.
We observe that from every state there are two options to the next step, such that at $\Delta = 3$ there were $\mathsf{N}=4= 2^{3-1}$ states.
We can generalize to $\mathsf{N}=2^{\Delta-1}$ states for the step $\Delta$.

At step $\Delta$ there are $\mathsf{N}$ states, covering the combinations of $j \geq 0$ is the number of innovations received and $\ell \geq 0$ is the number of dropouts where $\Delta=j+\ell+1$ to give
\begin{equation*}
    \pi_{j \ell} = p_i^j p_d^\ell p_d \pi_0 .
\end{equation*}
This shows the second part of the proposition and completes the proof.
\end{proof}

We now show an alternate proof to Theorem~\ref{thm:expectedlegitest} utilizing the structure of the Markov Chain to enumerate the sequences that give the possible estimation error covariances.

\begin{proof}
Proof of Theorem~\ref{thm:expectedlegitest}.

As $k\rightarrow\infty$, we can compute the expected estimation error of the legitimate estimator as the a possible expected estimation error conditioned on a given sequence multiplied by the probability of that sequence.
The total expectation of the estimation error of the legitimate estimator can be written as $k\rightarrow\infty$ is
\begin{equation*}
    E[\bar{P}_{k|k}] = \sum_{i=0}^\infty E[\bar{P}_{k|k} | \Delta = i] P(\Delta = i)
\end{equation*}
where $\Delta$ is the number of dropouts from the `in-sync' state.

Following Proposition~\ref{proposition:limitingdist}, we observe that there $\mathsf{N} = 2^{i-1}$ states at the step $\Delta = i$.
Thus at step $i$ we can write the conditioned expectation as a sum of the Markov Chain states at $\Delta=i$ multiplied by the probability of reaching that state
\begin{align*}
    &E[P_{k|k} | \Delta = i] P(\Delta = i) \\
    = &\sum_{m=0}^{\mathsf{N}} E[ P_{k|k} | j, \ell \wedge \Delta = i] P( j, \ell | \Delta = i) 
\end{align*}
where $m$ iterates through the combinations of $j$ and $\ell$.

Consider at the `in-sync' state $\Delta=0$, by Theorem~\ref{thm:mmseestimate} the estimation error covariance will be $P_{k|k} = 0$ and the probability is the limiting distribution of the first dropout $P(\Delta = 0) = \pi_0$
\begin{align*}
    E[P_{k|k} | \Delta = 0] P(\Delta = 0) = 0 \pi_0  = 0.
\end{align*}

Consider at step $\Delta=1$, the first dropout from `in-sync', the previous estimation error covariance is zero $P_{k-1|k-1} = 0$ so by Theorem~\ref{thm:mmseestimate} the estimation error covariance will be $P_{k|k} = Q$ and the probability is the limiting distribution of the first dropout $P(\Delta = 1) = \pi_D$ 
\begin{align*}
    E[P_{k|k} | \Delta = 1] P(\Delta = 0) = Q p_d \pi_0 .
\end{align*}
At step $\Delta=2$, there is a dropout and an innovation so the estimation error covariance will build slightly differently with the two limiting distributions
\begin{align*}
    &E[ P_{k|k} | 1, 0 \wedge \Delta = 1] P( 1,0 | \Delta = 2) \\&\quad= A Q A^\mathsf{T} p_i p_d \pi_0 \\
    &E[ P_{k|k} | 0,1 \wedge \Delta = 1] P( 0,1 | \Delta = 2) \\&\quad= (A Q A^\mathsf{T} + Q) p_d p_d \pi_0
\end{align*}
and together
\begin{align*}
    &E[ P_{k|k} | \Delta = 2] P( \Delta = 2) \\
    = &A Q A^\mathsf{T} p_i p_d \pi_0 + (A Q A^\mathsf{T} + Q) p_d p_d \pi_0 \\
    = &(A Q A^\mathsf{T}(p_i + p_d) + Q p_d) p_d \pi_0
\end{align*}

Consider at step $\Delta=3$ we find
\begin{align*}
    &E[ P_{k|k} | 2, 0 \wedge \Delta = 3] P( 2,0 | \Delta = 3) \\&\quad= A^2 Q A^{2\mathsf{T}} p_i^2 p_d \pi_0 \\
    &E[ P_{k|k} | 1, 1 \wedge \Delta = 3] P( 1,1 | \Delta = 3) \\&\quad= (A^2 Q A^{2\mathsf{T}} + Q) p_d p_i p_d \pi_0 \\
    &E[ P_{k|k} | 1, 1 \wedge \Delta = 3] P( 1,1 | \Delta = 3) \\&\quad= A(A Q A^\mathsf{T} + Q)A^\mathsf{T} p_i p_d p_d \pi_0 \\
    &E[ P_{k|k} | 0, 2 \wedge \Delta = 3] P( 0,2 | \Delta = 3) \\&\quad= (A(A Q A^\mathsf{T} + Q)A^\mathsf{T} + Q) p_d^2 p_d \pi_0
\end{align*}
which together give
\begin{align*}
    &E[ P_{k|k} | \wedge \Delta = 3] P( \Delta = 3) \\
    = &(A^2QA^{2\mathsf{T}} (p_i+p_d)(p_i+p_d) \\
    &+AQA^\mathsf{T}(p_i+p_d)p_d \\
    &+Qp_d(p_i+p_d)) p_d \pi_0 .
\end{align*}
The conditioned expectation multiplied by the limiting distributions reduce to the a combination of the dynamics by the probabilities of dropout and innovation, by the stationary distribution of the `in-sync' state.

We then observe that the sum for all steps after the `in-sync' state can be written as
\begin{align*}
    E[\bar{P}_{k|k}] &= \pi_0 p_d \left[ Q + \sum_{m=1}^\infty (p_i + p_d)^{m-1} \right. \\&\left.\times\left(\sum_{s=0}^{m-1} A^sQ (A^s)^\mathsf{T} p_d + A^{m}Q A^{m \mathsf{T}} (p_i+p_d)\right)\right] \\
    &= \left[\sum_{m=1}^\infty p_d (p_i+p_d)^{m-1} \sum_{s=0}^{m-1} A^sQ (A^s)^\mathsf{T} \right. \\ &\quad\left.+\sum_{m=1}^\infty (p_i+p_d)^m A^{m}Q A^{m \mathsf{T}} +Q \right] p_d \pi_0 \\
    &= \left[\sum_{m=1}^\infty p_d (p_i+p_d)^{m-1} \sum_{s=0}^{m-1} A^sQ (A^s)^\mathsf{T} \right. \\ &\quad\left.+\sum_{m=0}^\infty (p_i+p_d)^m A^{m}Q A^{m \mathsf{T}} \right] p_d \pi_0
\end{align*}
where we have two infinite sums.
Consider the second sum
\begin{align*}
    S &= \sum_{m=0}^\infty (p_i+p_d)^m A^{m}Q (A^m)^{\mathsf{T}} \\
    &= \sum_{m=0}^\infty (\sqrt{p_i+p_d}  A)^{m} Q (\sqrt{p_i+p_d} A^\mathsf{T})^{m} 
\end{align*}
which is a vector geometric series.
We can alternatively write $S$ as a sequence
\begin{align*}
    S_{m+1} = \sqrt{p_i+p_d}  A)^{m} S_{m} (\sqrt{p_i+p_d} A^\mathsf{T})^{m} + Q
\end{align*}
from $S_0 = Q$.
This equation is in the form of a Lyapunov Equation \cite{Anderson1979OptimalFiltering}, in the case that the matrix $\rho(\sqrt{p_i+p_d}  A) < 1$ by assumption of \eqref{eq:decisionvariableminbound} and $m \rightarrow \infty$ then $S_{m+1} = S_m = S$ and
\begin{equation*}
    S = (\sqrt{p_i+p_d}  A) S (\sqrt{p_i+p_d} A)^{\mathsf{T}} + Q
\end{equation*}
where $S$ is a unique stabilizing solution to the Lyapunov Equation.

Now consider the first sum
\begin{align*}
    S_1 &= \sum_{m=1}^\infty p_d (p_i+p_d)^{m-1} \sum_{s=0}^{m-1} A^s Q (A^s)^\mathsf{T} \\
    &= p_d \sum_{m=0}^\infty (p_i+p_d)^m \sum_{s=0}^{m} A^sQ (A^s)^\mathsf{T} \\
    &= p_d \sum_{m=0}^\infty (\sqrt{p_i+p_d} A)^sQ (\sqrt{p_i+p_d} A^\mathsf{T})^{m } \\&\qquad\times \sum_{s=0}^{\infty} (p_i+p_d)^s \\
    &= p_d S \sum_{s=0}^{\infty} (p_i+p_d)^s
\end{align*}
where we have two separable sums in geometric series.
We observe that the first part is the same as $S$ above and the second part is the standard scalar geometric series where
\begin{equation*}
     \sum_{s=0}^{\infty} (p_i+p_d)^s = \frac{1}{1 - (p_i + p_d)}
\end{equation*}
then
\begin{equation*}
    E[\bar{P}_{k|k}] = S \left( \frac{p_d}{1- (p_i + p_d)} + 1 \right) p_d \pi_0 .
\end{equation*}

Recalling that
$p_r = \mathbb{P}(\gamma_k = 1, \nu_k = 0) = \mu \mu_d$,
$p_i = \mathbb{P}(\gamma_k = 1, \nu_k = 1) = \mu (1-\mu_d)$,
and $p_d = \mathbb{P}(\gamma_k = 0) = (1-\mu)$,
and $\pi_0 = \frac{p_r}{1-p_i}$.
Then we can write
\begin{align*}
    E[\bar{P}_{k|k}] &= S \left( \frac{p_d}{1- (p_i + p_d)} + 1 \right) p_d \pi_0 \\
    &= S \left( \frac{(1-\mu)}{1- (\mu (1-\mu_d) + (1-\mu))} + 1 \right) (1-\mu) \pi_0 \\
    &= S \left( \frac{1-\mu}{1- \mu + \mu \mu_d - 1 + \mu} + 1 \right) (1-\mu) \frac{p_r}{1-p_i} \\
    &= S \left( \frac{1-\mu}{\mu \mu_d} + 1 \right) (1-\mu) \frac{\mu\mu_d}{1- \mu(1-\mu_d)} \\
    &= S \left( \frac{1-\mu}{\mu \mu_d} + \frac{\mu\mu_d}{\mu\mu_d} \right) \frac{\mu\mu_d}{1- \mu + \mu \mu_d} (1-\mu) \\
    &= S \frac{1-\mu+\mu\mu_d}{\mu \mu_d} \frac{\mu\mu_d}{1- \mu + \mu \mu_d} (1-\mu) \\
    &= S (1-\mu) 
\end{align*}
and the Lyapunov equation
\begin{align*}
    S &= (\sqrt{p_i+p_d}  A) S (\sqrt{p_i+p_d} A)^{\mathsf{T}} + Q \\
    S &= (\sqrt{1 - \mu \mu_d}  A) S (\sqrt{1 - \mu \mu_d} A)^{\mathsf{T}}  + Q .
\end{align*}

We then conclude the proof by stating the expected estimation error covariance as $k\rightarrow\infty$ as
\begin{equation*}
    E[P_{k|k}] = (1-\mu) S 
\end{equation*}
where
\begin{equation*}
    S = (\sqrt{1 - \mu \mu_d}  A) S (\sqrt{1 - \mu \mu_d} A)^{\mathsf{T}}  + Q .
\end{equation*}
This concludes the proof.
\end{proof}

\subsection{Eavesdropper MMSE}
\label{app:eavesestimate}

Proof of Lemma~\ref{lemma:classeavesdropperstateest}.
The following proof shows the estimation error covariance of an eavesdropper.

\begin{proof}
Consider a packet drop $\gamma_k^e = 0$ or the eavesdropper discards a successfully received packet $(\gamma_k^e,\nu_k,b_k) = (1,0,0)$ or $(\gamma_k^e,\nu_k,b_k) = (1,1,0)$, no new information is received, so the state estimate is the prediction from the previous state.
Following the dynamics and applying \eqref{eq:eavesdropperMMSE} the state estimate is
$\hat{x}_k^e = E[x_k | \mathcal{I}_{k-1}^e] = A \hat{x}_{k-1}^e$,
with covariance
\begin{align*}
    P_{k|k-1}^e &= E[(x_k - \hat{x}_{k|k-1}^e)(x_k - \hat{x}_{k|k-1}^e)^\mathsf{T} | \mathcal{I}_{k-1}^e ] \\
    &= A P_{k-1|k-1}^e A^\mathsf{T} + Q
\end{align*}
recalling that $w_k$ and $x_k$ are uncorrelated.

Consider a successful packet receipt of the state transmission, such that $z_k = x_k$, which the eavesdropper uses $(\gamma_k^e,\nu_k,b_k) = (1,0,1)$.
The transmission is directly used as the state estimate $\hat{x}_{k|k}^e = z_k = x_k$, with covariance
\begin{align*}
    P_{k|k}^e &= E[(x_k - \hat{x}_{k|k}^e)(x_k - \hat{x}_{k|k}^e)^\mathsf{T} | \mathcal{I}_k^e]
    = 0 .
\end{align*}

Finally, consider a successful packet receipt of the innovation, such that $z_k = x_k - A x_{k-1} + \chi_k$, which the eavesdropper mistakenly believes is the state and uses $(\gamma_k^e,\nu_k,b_k) = (1,1,1)$.
The transmission is directly used as the state estimate
\begin{align*}
    \hat{x}_{k|k}^e = z_k = x_k - A x_{k-1} + \chi_k = w_{k-1} + \chi_k.
\end{align*}
The covariance is then 
\begin{align*}
    &P_{k|k}^e = E[(x_k - \hat{x}_{k|k}^e) (x_k - \hat{x}_{k|k}^e)^\mathsf{T} | \mathcal{I}_k^e] \\
    &= E[(x_k - w_{k-1} - \chi_k) (x_k - w_{k-1} - \chi_k)^\mathsf{T} ] \\
    &= E[(A x_{k-1} + w_{k-1} - w_{k-1} - \chi_k) \\&\qquad\times (A x_{k-1} + w_{k-1} - w_{k-1} - \chi_k)^\mathsf{T}] \\
    &= E[(A x_{k-1} - \chi_k) (A x_{k-1} - \chi_k)^\mathsf{T} ] \\
    &= A E[x_{k-1} x_{k-1}^\mathsf{T} ] A^\mathsf{T} - E[\chi_k x_{k-1}^\mathsf{T}] A^\mathsf{T} \\&\qquad- A E[x_{k-1} \chi_k^\mathsf{T}] + E[\chi_k \chi_k^\mathsf{T}]
\end{align*}
From \cite{Anderson1979OptimalFiltering} the covariance of $x_{k-1}$
\begin{equation*}
    E[x_{k-1} x_{k-1}^\mathsf{T}] = A^{k-1} \Sigma_0 (A^{k-1})^\mathsf{T} + \sum_{\ell=0}^{k-2} A^{k-2-\ell} Q (A^{k-2-\ell})^\mathsf{T} ,
\end{equation*}
that the additive noise is designed \eqref{eq:additivenoisecovariancedesign} such that
\begin{equation*}
    E[\chi_k\chi_k^\mathsf{T}] = A^k \Sigma_0 (A^k)^\mathsf{T} + \sum_{\ell = 0}^{k-2} A^{k-1-\ell} Q (A^{k-1-\ell})^\mathsf{T} ,
\end{equation*}
and $\chi_k$ and $x_{k-1}$ are uncorrelated, $E[x_{k-1} \chi_k^\mathsf{T}] = 0$,
then
\begin{align*}
    P_{k|k}^e 
    &= 2\left(A^k \Sigma_0 (A^k)^\mathsf{T} + \sum_{\ell=0}^{k-2} A^{k-1-\ell} Q (A^{k-1-\ell})^\mathsf{T}\right) .
\end{align*}
This shows the eavesdropper's state estimate and associated covariance and completes the proof.
\end{proof}

\subsection{Eavesdropper Expected Estimation Error Covariance}
\label{app:classeavesestimateperformance}
Proof of Theorem~\ref{thm:classeavesdropperexpectation}.
The following proof shows the expected estimation error covariance of the eavesdropper 

\begin{proof}
    We are going to show via proof by induction that the expected estimation error of the eavesdropper for time $k>0$ can be written as a sum of the sequence of dropouts and encoded innovations from the first transmission, we repeat the form of $E[P_{k|k}^e]$ from Theorem~\ref{thm:classeavesdropperexpectation}
    \begin{align*}
        &E[P_{k|k}^e] = (p_d^e)^k A^{k-1} \Sigma_0 (A^{k-1})^\mathsf{T} 
        + \sum_{\ell=0}^{k-1} (p_d^e)^{\ell+1} A^\ell Q (A^\ell)^\mathsf{T} \\ 
        + &p_i^e \sum_{\ell=0}^{k-1} (p_d^e)^\ell 2 \Bigl(A^k \Sigma_0 (A^k)^\mathsf{T} +  \sum_{j=0}^{k-\ell-2} A^{k-1-j} Q (A^{k-1-j})^\mathsf{T} \Bigr) .
    \end{align*}
    It is helpful for the proof to rewrite $E[P_{k|k}^e]$ as
    \begin{align}
        E[P_{k|k}^e] &= (p_d^e)^k A^{k-1} \Sigma_0 (A^{k-1})^\mathsf{T} + \sum_{\ell=0}^{k-1} (p_d^e)^{\ell+1} A^\ell Q (A^\ell)^\mathsf{T} \nonumber \\ 
        &+p_i^e \sum_{\ell=0}^{k-1} (p_d^e)^\ell A^\ell f_{k-\ell} (A^{\ell})^\mathsf{T} . \label{eq:proofinductioneavesdropclass}
    \end{align}
    where the expected estimation error covariance on use of an innovation $(\gamma_k^e,\nu_k,b_k) = (1,1,1)$ at time $i > 1$ from the result in Lemma~\ref{lemma:classeavesdropperstateest} as
    \begin{align*}
        f_i = 2\left(A^i \Sigma_0 (A^i)^\mathsf{T} + \sum_{j=0}^{i-2} A^{i-1-j} Q (A^{i-1-j})^\mathsf{T} \right) .
    \end{align*}

    We show \eqref{eq:proofinductioneavesdropclass} via proof by induction.

    An eavesdropper has three possible outcomes:
    drop or discards a packet $(\varphi_k=1)$ with probability $p_d^e = \mathbb{P}(\gamma_k^e=0) + \mathbb{P}(\gamma_k^e=1,\nu_k=0,b_k=0) + \mathbb{P}(\gamma_k^e=1,\nu_k=1,b_k=0)$,
    receive and use a state packet $(\varphi_k=2)$ with probability $p_r^e = \mathbb{P}(\gamma_k^e=1,\nu_k=0,b_k=1)$, and
    receive and use an encoded innovation packet $(\varphi_k=3)$ with probability $p_i^e = \mathbb{P}(\gamma_k^e=1,\nu_k=1,b_k=1)$.

    Consider $k=0$ from the definition $E[P_{0|0}^e] = \Sigma_0 p_d^e + 0 (p_r^e + p_i^e)$ as only the state is transmitted at the first instance, $z_0 = x_0$.
    Consider the first time $k=1$ from the definition
    \begin{align*}
        E[P_{1|1}^e] &= \sum_{y=1}^3 E[P_{1|1}^e | \varphi_1 = y] \mathbb{P}(\varphi_1 = y) \\
        &= (p_d^e)^2 A \Sigma_0 A^\mathsf{T} + p_i^e f_1 + 0 p_r^e .
    \end{align*}

    We now show that if \eqref{eq:proofinductioneavesdropclass} holds for time $k$, then the form \eqref{eq:proofinductioneavesdropclass} also holds for time $k+1$.
    From the result in Lemma~\ref{lemma:classeavesdropperstateest}
    \begin{align*}
        E[P_{k+1|k+1}^e] &= \sum_{y=1}^3 E[P_{k+1|k+1}^e | \varphi_{k+1} = y] \mathbb{P}(\varphi_{k+1} = y) \\
        &= p_d^e (A E[P_{k|k}^e] A^\mathsf{T} + Q) + p_i^e f_{k+1} + p_r^e 0 
    \end{align*}
    By the proposed form for $E[P_{k|k}^e]$ in \eqref{eq:proofinductioneavesdropclass}
    \begin{align*}
        &E[P_{k+1|k+1}^e] = p_d^e (A E[P_{k|k}^e] A^\mathsf{T} + Q) + p_i^e f_{k+1} \\
        &= p_d^e A \left( (p_d^e)^k A^{k-1} \Sigma_0 (A^{k-1})^\mathsf{T} + \sum_{\ell=0}^{k-1} (p_d^e)^{\ell+1} A^\ell Q (A^\ell)^\mathsf{T} \right. \\ 
        &\quad\left.+p_i^e \sum_{\ell=0}^{k-1} (p_d^e)^\ell A^\ell f_{k-\ell} (A^{\ell})^\mathsf{T} \right)A^\mathsf{T} + p_d^e Q + p_i^e f_{k+1} \\
        &= (p_d^e)^{k+1} A^{k} \Sigma_0 (A^{k})^\mathsf{T} + \sum_{\ell=0}^{k} (p_d^e)^{\ell+1} A^\ell Q (A^\ell)^\mathsf{T} \\ 
        &\quad +p_i^e \sum_{\ell=0}^{k} (p_d^e)^\ell A^\ell f_{k-\ell} (A^{\ell})^\mathsf{T} A^\mathsf{T}
    \end{align*}
    which is the form \eqref{eq:proofinductioneavesdropclass} at time $k+1$.

    This completes the proof.
\end{proof}

\subsection{Naive and Suspicious Eavesdropper Expected Estimation Error Covariance}
\label{app:naivesusexpected}
Proof of Corollary~\ref{corr:naiveeavesexpectation} and Corollary~\ref{corr:suspiciouseavesexpectation}.
The following proof shows the expected estimation error covariance of the naive and suspicious eavesdroppers.

\begin{proof}
    The naive eavesdropper utilizes any transmission that it successfully receives.
    The probabilities of the three outcomes for the naive eavesdropper are:
    standard dropout $p_d^e = 1-\mu_e$, successful receipt of the state $p_r^e = \mu_e \mu_d$, and successful receipt of an innovation $p_i^e = \mu_e (1-\mu_d)$.

    The suspicious eavesdropper utilizes a successfully received transmission with random chance based on the type of transmission it receives.
    The probabilities of the three outcomes for the suspicious eavesdropper are:
    standard dropout or discard $p_d^e = 1 - \mu_e \mu_d \mu_b - \mu_e \bar{\mu}_b + \mu_e \mu_d \bar{\mu}_b$,
    successful receipt of the state $p_r^e = \mu_e \mu_d \mu_b$,
    and successful receipt of an innovation $p_i^e = \mu_e (1-\mu_d) \bar{\mu}_b$.

    Applying $p_d^e$ and $p_i^e$ for both the naive and suspicious eavesdroppers to Theorem~\ref{thm:classeavesdropperexpectation}, we inspect the resulting terms.
    Under assumption that $\rho(\sqrt{p_d^e}A)<1$, then for the first two terms as $k\rightarrow\infty$
    \begin{align*}
        &(p_d^e)^k A^{k-1} \Sigma_0 (A^{k-1})^\mathsf{T} \rightarrow 0\\
        &\sum_{\ell=0}^{k-1} (p_d^e)^{\ell+1} A^\ell Q (A^\ell)^\mathsf{T} \rightarrow S^{e,n}
    \end{align*}
    where $0$ is a zero matrix of appropriate size and $S^{e,n}$ is the converged stabilizing solution to the Lyapunov equation.
    Otherwise $S^{e,n}$ is undefined, and both terms diverge.

    Let us inspect the two parts of the last term of Theorem~\ref{thm:classeavesdropperexpectation} as $k\rightarrow\infty$
    \begin{align*}
        &\textrm{trace}~A^k \Sigma_0 (A^k)^\mathsf{T} \rightarrow \infty , \quad \textrm{if~} \rho(A) > 1 \\
        \textrm{or} ~ &\textrm{trace}~A^k \Sigma_0 (A^k)^\mathsf{T} > \min_i \lambda_i(A \Sigma_0 A^\mathsf{T}) , \quad \textrm{if~} \rho(A) = 1,
    \end{align*}
    where $\min_i \lambda_i(A \Sigma_0 A^\mathsf{T})$ is the minimum eigenvalue of $A \Sigma_0 A^\mathsf{T}$,
    and the second part of the last term
    \begin{align*}
        &\textrm{trace}~\sum_{\ell=0}^{k-2} A^{k-1-\ell} Q (A^{k-1-\ell})^\mathsf{T}  \rightarrow \infty .
    \end{align*}
    By assumption that the pair $(A,\sqrt{Q})$ is controllable, there are no eigenvectors of $A$ in the nullspace of $\sqrt{Q}$.
    In the case that $\rho(A) = 1$, the eigenvector of $A$ associated with the eigenvalue on the unit circle extracts a combination of the eigenvalues of $\sqrt{Q}$, and remains non-zero as $k \rightarrow \infty$.
    Thus we conclude that as $k\rightarrow\infty$ then we have an infinite sum of non-zero eigenvalues of $Q$.
    
    The expectation of the eavesdropper's estimation error diverge to infinity, or $\textrm{trace}~E[P_{k|k}^e] \rightarrow \infty$, such that both the naive and suspicious eavesdroppers have an unbounded estimation error satisfying condition (ii) of Definition~\ref{definition:perfectsecrecy}.
    This completes the proof.
\end{proof}

\subsection{Smart Eavesdropper Expected Estimation Error Covariance}
\label{app:smarteavsexpected}
Proof of Lemma~\ref{lemma:smarteavesexpectedest}.
The following proof shows the expected estimation error of the smart eavesdropper.

\begin{proof}
The smart eavesdropper has two outcomes: 
standard dropout or discard of innovation with probability $p_d^e = 1-\mu_e\mu_d$, or successful receipt of a state estimate with probability $p_r^e = \mu_e \mu_d$.
In this proof we assume that \eqref{eq:eavesdroppernetworkcorrect} holds such that the eavesdropper has a bounded estimate.

The expected estimation error covariance of the smart eavesdropper at time $k$ can be found from Theorem~\ref{thm:classeavesdropperexpectation}
\begin{align*}
        E[P_{k|k}^e]
        = &(1-\mu_e\mu_d)^k A^{k-1} \Sigma_0 (A^{k-1})^\mathsf{T} \\&+ \sum_{\ell=0}^{k-1} (1-\mu_e\mu_d)^{\ell+1} A^\ell Q (A^\ell)^\mathsf{T} . 
    \end{align*}
Consider the first term of the initial estimation error term with $\Sigma_0$, by assumption of \eqref{eq:eavesdroppernetworkcorrect} then $\rho(\sqrt{1-\mu_e\mu_d} A) < 1$, and we observe that as $k\rightarrow\infty$ then $\left(\sqrt{1-\mu_e\mu_d} A\right)^k \rightarrow 0$, and the initial estimation error will be exponentially forgotten.
Consider the second term of the sum to $k-1$, 
\begin{align*}
    S_k^e &= \sum_{\ell=0}^{k-1} (1-\mu_e \mu_d)^\ell A^j Q (A^\mathsf{T})^\ell
\end{align*}
which can be written as a Lyapunov equation from $S_0 = Q$
\begin{align*}
    S_k^e &= \sqrt{1-\mu_e \mu_d} A S_{k-1}^e A^\mathsf{T} \sqrt{1-\mu_e \mu_d} + Q  .
\end{align*}
The stabilized solution $S^e$ can be found by taking $k \rightarrow \infty$ or setting $S_{k-1} = S_k = S^e$ and solving for the unique stabilizing solution $S^e$
\begin{align*}
    S^e &= \sqrt{1-\mu_e \mu_d} A S^e A^\mathsf{T} \sqrt{1-\mu_e \mu_d} + Q .
\end{align*}
The expected estimation error of the smart eavesdropper is
\begin{align*}
    E[P_{k|k}^e] &= (1-\mu_e\mu_d) S^e .
\end{align*}
We note the performance is as expected of a remote state estimator transmitting the state every time instance with channel quality $p_r^e = \mu_e\mu_d$.
This completes the proof.
\end{proof}

\subsection{Monotonicity of Lyapunov Equation}
\label{app:lyapeqmonotonicity}

Proof of Lemma~\ref{lemma:monotonicityresult}.
The following proof shows a monotonicity result on the scaling coefficient on the Lyapunov equation.

\begin{proof}
    Consider a $\beta^\star$ and $\beta$ where $0 < \beta , \beta^\star < 1$ where $\rho(\sqrt{1-\beta} A) < 1$ and $\rho(\sqrt{1-\beta^\star} A) < 1$
    and introduce two Lyapunov equations as stabilizing recursions \cite{Anderson1979OptimalFiltering}
    \begin{align*}
        W_{k+1} = \sqrt{1 - \beta} A &W_{k} A^\mathsf{T} \sqrt{1-\beta} + Q, \\ 
        W_{k+1}^\star = \sqrt{1-\beta^\star} A &W_{k}^\star A^\mathsf{T} \sqrt{1-\beta} + Q 
    \end{align*}
    with $W_{0}^\star = W_{0} = Q$, which converge to the unique-stabilizing solutions $W$ and $W^\star$, respectively.
    Let us introduce $\alpha = \sqrt{1- \beta^\star}/\sqrt{1-\beta}$ and $\tilde{A} = \sqrt{1-\beta} A$, and note that $\rho(\tilde{A}) < 1$ and $\rho(\alpha \tilde{A}) < 1$.
    
    Consider the case that $\beta^\star < \beta$ then $\alpha > 1$.
    The two Lyapunov equations can be written as
    \begin{align*}
        W_{k+1} &= \tilde{A} W_k \tilde{A}^\mathsf{T} + Q , \quad \textrm{and} \quad
        W_{k+1}^\star = \alpha \tilde{A} W_k^\star \tilde{A}^\mathsf{T} \alpha + Q .
    \end{align*}
    Let us introduce the difference $V_k = W_k^\star - W_k$, which can written as a function of the previous difference
    \begin{equation}
        V_k = (\alpha^{2k} - 1) \tilde{A}^k Q (\tilde{A}^\mathsf{T})^k + V_{k-1}
        \label{eq:proofdifferenceVk}
    \end{equation}
    from $V_0 = 0$.
    We show \eqref{eq:proofdifferenceVk} via proof by induction.
    Let us first evaluate at $k=0$ and $k=1$ 
    \begin{align*}
        V_{0} &= W_0^\star - W_0 = Q - Q = 0, \quad \textrm{and} \\
        V_{1} &= W_1^\star - W_1 = \alpha \tilde{A} W_0^\star \tilde{A}^\mathsf{T} \alpha + Q - \tilde{A} W_0 \tilde{A}^\mathsf{T} - Q \\
        &= (\alpha^2 - 1) \tilde{A} Q \tilde{A}^\mathsf{T} + V_0 .
    \end{align*}
    Let us assume the form \eqref{eq:proofdifferenceVk} and show the form at $k+1$ from the definition of $W_k^\star$ and $W_k$,
    \begin{align*}
        V_{k+1} &= W_{k+1}^\star - W_{k+1} \\
        &= \sum_{j=0}^{k+1} (\alpha \tilde{A})^j Q  (\tilde{A}^\mathsf{T} \alpha)^j - \sum_{\ell=0}^{k+1} \tilde{A}^\ell Q (\tilde{A}^\mathsf{T})^\ell \\
        &= (\alpha^{2(k+1)} - 1) \tilde{A}^{k+1} Q (\tilde{A}^\mathsf{T})^{k+1} \\&\quad+ \sum_{j=0}^{k} (\alpha \tilde{A})^j Q  (\tilde{A}^\mathsf{T} \alpha)^j - \sum_{\ell=0}^{k} \tilde{A}^\ell Q (\tilde{A}^\mathsf{T})^\ell \\
        &= (\alpha^{2(k+1)} - 1) \tilde{A}^{k+1} Q (\tilde{A}^\mathsf{T})^{k+1} + V_k
    \end{align*}
    which produces the form \eqref{eq:proofdifferenceVk} at iteration $k+1$.

    We now explore the trace of $V_k$.
    \begin{align*}
        &\textrm{trace~} V_{k} = \textrm{trace~} \left((\alpha^{2 k} - 1) \tilde{A}^{k} Q (\tilde{A}^\mathsf{T})^{k} + V_{k-1} \right) \\
        &= (\alpha^{2 k} - 1) \textrm{trace} \left( \tilde{A}^{k} Q (\tilde{A}^\mathsf{T})^{k} \right) + \textrm{trace~} V_{k-1} .
    \end{align*}
    We observe that $\textrm{trace~} (\tilde{A} Q \tilde{A}^\mathsf{T}) > 0$ as the pair $(A,\sqrt{Q})$ is controllable.
    By definition $\alpha > 1$ so it follows that $\alpha^{2 j} - 1 > 0$ for all $j>0$.
    Thus the first term is strictly positive
    \begin{align*}
        (\alpha^{2 k} - 1) \textrm{trace} \left( \tilde{A}^{k} Q (\tilde{A}^\mathsf{T})^{k} \right) > 0 .
    \end{align*}
    Consider the trace of $V_{1}$ using the same properties as above %
    \begin{align*}
        \textrm{trace~} V_{1} 
        &= (\alpha^2 - 1) \textrm{trace~} (\tilde{A} Q \tilde{A}^\mathsf{T}) > 0 .
    \end{align*}
    At $k=2$, then trace $V_{k-1} = $ trace $V_1 > 0$, and trace $V_2 > 0$.
    Following a proof by induction argument, we conclude that trace $V_k > 0$ for $k>0$.
    This implies that at the difference in stabilized Lyapunov equation solutions $\textrm{trace~}(W^\star - W) > 0$, and that $\textrm{trace~}W^\star > \textrm{trace~}W$.
    This concludes the proof.
\end{proof}

\balance

\bibliographystyle{IEEEtran}
\bibliography{IEEEabrv,ref}

\end{document}